\documentclass[12pt]{article}   	
\usepackage{graphicx}
\usepackage[numbers]{natbib}				
\usepackage{amssymb,amsmath,amsthm}
\usepackage{bbm,csvsimple}
\usepackage{chngcntr}
\usepackage{setspace, enumerate,array}
\usepackage[colorlinks,allcolors=blue]{hyperref}
\usepackage{subcaption}
\usepackage{csvsimple}
\usepackage[table,svgnames]{xcolor}
\usepackage[font={small}]{caption}
\usepackage{floatrow,multirow}
\usepackage{hhline}
\usepackage{url}
\usepackage{comment}
\newcommand{\blind}{0}

\newfloatcommand{capbtabbox}{table}[][\FBwidth]

\theoremstyle{plain}
\newtheorem{assump}{Assumption}

\newcolumntype{P}[1]{>{\centering\arraybackslash}p{#1}}

\theoremstyle{remark}
\newtheorem*{remark}{Remark}

\theoremstyle{plain}
\newtheorem{theorem}{Theorem}

\theoremstyle{plain}

\newtheorem{corollary}{Corollary}[theorem]

\begin{document}

\def\spacingset#1{\renewcommand{\baselinestretch}%
{#1}\small\normalsize} \spacingset{1}

\if0\blind
{
  \title{\bf Adaptive Ridge-Penalized Functional Local Linear Regression}
  \author{Wentian Huang\hspace{.2cm}\\
    Department of Statistics and Data Science, Cornell University\\
    and \\
    David Ruppert \\
    Department of Statistics and Data Science, Cornell University}
  \maketitle
} \fi

\bigskip
\begin{abstract}
We introduce an original method of multidimensional ridge penalization in functional local linear regressions. The nonparametric regression of functional data is extended from its multivariate counterpart, and is known to be sensitive to the choice of $J$, where $J$ is the dimension of the projection subspace of the data. Under multivariate setting, a roughness penalty is helpful for variance reduction. However, among the limited works covering roughness penalty under the functional setting, most only use a single scalar for tuning. Our new approach proposes a class of data-adaptive ridge penalties, meaning that the model automatically adjusts the structure of the penalty according to the data sets. This structure has $J$ free parameters and enables a quadratic programming search for optimal tuning parameters that  minimize the estimated mean squared error (MSE) of prediction, and is capable of applying different roughness penalty levels to each of the $J$ basis.

The strength of the method in prediction accuracy and variance reduction with finite data is demonstrated through multiple simulation scenarios and two real-data examples. Its asymptotic performance is proved and compared to the unpenalized functional local linear regressions.
\end{abstract}

\noindent%
\vfill

\newpage
\spacingset{1.75}

\section{Introduction} \label{section: introduction}
Functional data analysis has received increasing attention during the past few decades with applications in a variety of fields such as chemometrics, medicine, and environmental science. This article focuses on scalar-on-function regression where an unknown function, $m$, describes the relationship between a predictor function $X$ in some Hilbert space and a real scalar $Y$.  The model is $Y = m(X) + \epsilon$ where $\epsilon$ is random error.   We assume an independent, identically distributed sample $(X_i,Y_i)$, $i = 1,\dots,n$.

Past work such as Cai et al.\ (2006 \cite{cai2006}) and Reiss and Ogden (2007 \cite{reiss2007}) discussed estimation when $m$ is linear, so that $m(X) = \langle X, \beta\rangle$, the inner product of $X$ and an unknown coefficient function $\beta$.   However,  the linearity assumption often fails to hold. For instance, in Section \ref{sec:realdata} we plot the estimated derivatives of $m$ at each observed function, $X_i$, in two real data sets (See Fig.\ \ref{fig:truck_sample} \& Fig.\ \ref{fig:corn_sample}). Linearity implies that the derivative of $m$ at $X$ is equal to $\beta$ at all $X$.  Variation of the derivative's shape as $X_i$ varies shows the nonlinearity of $m$ in these examples. 

Nonparameteric methods that have been widely used in multivariate regression have been extended to functional predictors and have shown strong performance there. For example, Ferraty et al.\ (2007 \cite{ferraty2007}) applied the well known Nadaraya-Watson kernel estimator to regression with functional predictors.  For regression with a scalar or low-dimensional covariate, local polynomial regression has advantages over kernel regression, e.g., better behavior near the boundary of the covariate space \citep{fan1992design,ruppert1994}.  Therefore, it is natural to study local polynomial functional regression.

Ba\'illo and Gran\'e (2009 \cite{baillo2009}) first extended the multivariate local linear regression estimator of Ruppert and Wand (1994 \cite{ruppert1994}), where $X$ is finite-dimensional, to functional data. Boj et al.\ (2010 \cite{boj2010distance}) and Barrientos-Marin et al.\ (2010 \cite{barrientos2010locally}), among many others, also studied local linear regression with functional predictors. Ferraty and Nagy (2019 \cite{ferraty2019}) discussed in detail the implementation of functional local linear regression (FLLR) and its asymptotic behavior, while Berlinet et al.\ (2011 \cite{berlinet2011}) explored the model from a purely theoretical perspective.  As in multivariate regression,  FLLR often has better prediction accuracy than kernel estimators. 

Because of the so-called curse-of-dimensionality, nonparametric estimators such as local linear regression can be problematic in high dimensional spaces.  Function spaces are infinite dimensional, so local polynomial regression might seem unsuitable for functional regression.  Fortunately, functional data often lie in a low-dimensional subspace, e.g., in the space spanned by the first few principal component directions.  Therefore, to implement local polynomial functional regression, one can project the data onto a subspace of dimension~$J$, e.g., the first $J$ principal components, where $J$  is a tuning parameter.  However, the estimator can be sensitive to the choice of $J$ and, even with the best choice of $J$, the estimator will likely be improved with a roughness penalty.

To improve the FLLR estimator, we propose a data-adaptive ridge roughness penalty.
The most general ridge penalty matrix is a $J \times J$ positive semidefinite matrix.  Data-based selection of this type of penalty matrix with $J(J+1)/2$ free parameters can be difficult and can result in an unstable and inefficient estimator.
Therefore, we propose a data-adaptive ridge penalization that utilizes a specific class of positive semidefinite diagonalizable matrices. As will be
 shown later, this structure with only $J$ free parameters enables a quadratic programming search for optimal tuning parameters that minimize the estimated mean squared error (MSE) of prediction. Our method of penalization also accommodates a different roughness penalty level on each basis function and avoids the computational cost of multivariate cross validation as $J$ increases. 

 Reiss and Ogden (2007 \cite{reiss2007})  suggested a univariate roughness penalty for functional linear models.  Reiss et al.\ (2017 \cite{reiss2017}) explored adding a fixed univariate ridge penalty onto nonparametric functional estimators with smoothing splines. Both papers select a single smoothing parameter by generalized cross validation or restricted maximum likelihood (REML) estimation of variances, and neither discussed the estimators' asymptotic behaviors.  As far as we know, there is no previous work investigating multidimensional ridge penalties in functional  nonparametric regression. Our estimator has strong prediction performance in both simulations and real data examples, especially when the model is nonlinear. In addition, the method shows effective bandwidth size control for finite data samples, proving its strength in variance reduction. Asymptotic properties of the new estimator are derived, and a detailed implementation is provided, including a two-step bandwidth selection for estimating $m$ and its  functional derivative $m'$.

In Section \ref{section: model}, we introduce our model and the design of the ridge penalty. In Section \ref{sec: estimation}, we estimate the mean square error (MSE) of our estimator and discuss its asymptotic estimator behavior. Section \ref{sec:sim} provides a detailed description of the implementation of the estimator and includes a comprehensive simulation study to compare the performance of multiple nonparametric methods under different linearity levels of $m$. Section \ref{sec:realdata} uses two real datasets to examine the performance of our method. In the end, we discuss potential future work. Additional results and detailed proofs can be found in the supplementary materials.

\section{Methodology} \label{section: model}
\subsection{Functional Local Linear Regression} \label{sec:fllr}
We consider a pair of variables $(X, Y) \in \mathcal{L}^2(\mathcal{T}) \times \mathbb{R}$, which means $X$ is a square integrable random function over a compact interval $\mathcal{T}$, and $Y$ is real valued. Suppose there exists a regression model $Y = m\left(X\right) + \epsilon$, where $m:  \mathcal{L}^2(\mathcal{T}) \to \mathbb{R}$ is first order differentiable, and $\epsilon \sim N\left(0, \sigma_e^2\right)$. In this article we are interested in the estimation of $E\left(Y|x\right)=m(x)$ at a point $x$, using $n$ i.i.d.\ samples $(X_i, Y_i)$ collected from the joint distribution $(X, Y)$. 

As discussed in the introduction, we use functional local linear regression (FLLR) to estimate $m(x)$.  However, FLLR estimates not only $m(x)$ but also its first derivative, $m'_x:  \mathcal{L}^2(\mathcal{T}) \rightarrow \Re$.  Although FLLR has become a well-studied technique for nonparametric functional regression, there has been relatively little research on regularizing the high-dimensional estimate of $m'_x$. For this purpose, we suggest a new FLLR model with data-adaptive ridge penalization, which we denote as FLLR-r.  As will be seen in Section~\ref{sec:sim}, regularization of $\hat m'_x$ also improves $\hat m(x)$.

Below are several assumptions needed for the FLLR-r estimator:

\begin{assump} \label{A1}
The continuously differentiable function $K: \mathbb{R} \to \mathbb{R}^+$ is a kernel of Type I, whose definition can be found in, for example, Ferraty and Vieu (2006 \cite{ferraty2006nonparametric}): $\int K=1$, and $c_K \mathbbm{1}_{[0, 1]} \le K \le C_K \mathbbm{1}_{[0, 1]}$ where $c_K, C_K >0$;
\end{assump}

\begin{assump} \label{A2}
$\forall h > 0, \psi_x(h)=P(\|X_i -x\| < h) > 0$. Also, as $n \to \infty$, $h = h_n \to 0$, $n\psi_x(h) \to \infty$, $J = J\left(n\right) \to \infty$. 
\end{assump}

\begin{assump} \label{A3}
For $x \in \mathcal{L}^2 \left(\mathcal{T}\right)$, the model $m: \mathcal{L}^2 (\mathcal{T}) \to \mathbb{R}$ is first and second order differentiable at its neighborhood $\mathcal{N}_x$, with the corresponding bounded derivative functional $m'_x: \mathcal{L}^2 (\mathcal{T}) \to \mathbb{R}$ and $m''_x: \mathcal{L}^2 (\mathcal{T}) \times \mathcal{L}^2 (\mathcal{T}) \to \mathbb{R}$. Also, for all $u \in \mathcal{L}^2 (\mathcal{T})$ and $x+u \in \mathcal{N}_x$, there is $0 <\rho<1$ and $r = x+\rho u$ s.t.\ 
\begin{equation*}
    m(x+u) = m(x) + m'_x(u) + \dfrac{1}{2} m''_r(u^2)
\end{equation*}
\end{assump}
Assumption \ref{A3} is an application of Taylor's Theorem in function spaces (Zeidler, 1995 \cite{zeidler1995}), and is similar to Assumption H1 in Ferraty and Nagy (2019 \cite{ferraty2019}). 

Let $\phi_1, \phi_2, \ldots \in \mathcal{L}^2(\mathcal{T})$ be a set of orthogonal basis functions, $X = \sum_{j=1}^{\infty} \langle X, \phi_j \rangle \phi_j$, $c_{ij} = \langle X_i-x, \phi_j \rangle$, and $X_i-x=\sum_{j=1}^{\infty} c_{ij} \phi_j$. With a first-order Taylor expansion at $x$, we have
\begin{align} \label{eq: taylor}
E(Y_i|X_i)=m(X_i) & =m(x)+m'_x\left(X_i-x\right)+o_p\left(\|X_i-x\|\right).
\end{align}
Eq.~(\ref{eq: taylor}) can therefore be estimated using a basis truncated at $J$: 
\begin{equation*}
m(X_i) \approx m(x) + \sum_{j=1}^{\infty} c_{ij} m'_x\left(\phi_j\right) \approx m(x)+\mathbf{c}_i^T\mathbf{m}'_{x, J},    
\end{equation*}
where $\mathbf{c}_i=(c_{i1}, \ldots, c_{iJ})^T$ and $\mathbf{m}'_{x, J}=\left \{m'_x(\phi_1), \ldots, m'_x(\phi_J)\right\}^T$. Define $\boldsymbol{\beta}=(\beta_1, \ldots, \beta_J)^T=\mathbf{m}'_{x, J}$, so $\boldsymbol{\beta}^T \boldsymbol{\Phi}$ is the derivative functional $m'_x$ projected on the subspace spanned by $\boldsymbol{\Phi} = \{ \phi_1, \ldots, \phi_J\}^T$.
In addition, $\mathbf{C}$ is an $n \times J$ matrix with rows $\mathbf{c}_i^T$, and $\mathbf{Y}$ is the vector of responses $Y_i$. FLLR then estimates $m(x)$ and $\boldsymbol{\beta}$ by minimizing the weighted sum of squared errors $\sum_{i=1}^n \left(Y_i - m(x) - \mathbf{c}_i^T\boldsymbol{\beta}\right)^2 \Delta_i$,
with the kernel weights $\Delta_i=\dfrac{K\left(\|X_i-x\|/h\right)}{E\left\{K\left(\|X_i-x\|/h\right)\right\}}$.

\subsection{Ridge Penalty in FLLR} \label{sec:ridge}
By its nature, an FLLR model is characterized by the truncated basis count $J$ and the bandwidth $h$, which in practice are usually determined by cross validation. As mentioned earlier, we'd like to introduce a multidimensional ridge penalty into the functional regression. The new method constructs the penalty using a data-adaptive basis learnt from sample functions, and enables a parameter selection that minimizes the finite sample estimation error of $m(x)$.

Let $\mathbf{H}^*$ be a $J \times J $ positive semidefinite penalty matrix, and $\mathbf{H}=\begin{pmatrix} 0 & \mathbf{0}^T \\ \mathbf{0} & \mathbf{H}^*\end{pmatrix}$. Optimal estimates of $m(x)$ and $\boldsymbol{\beta}$ satisfy the following:
\begin{equation} \label{eq: ridgerule}
\hat{m}(x), \hat{\boldsymbol{\beta}} = \text{argmin}_{m(x), \boldsymbol{\beta}} \left \{\sum_{i=1}^n \left(Y_i - m(x) - \mathbf{c}_i^T\boldsymbol{\beta}\right)^2 \Delta_i + \boldsymbol{\beta}^T \mathbf{H}^* \boldsymbol{\beta} \right \}.
\end{equation}
Thus, 
\begin{equation} \label{eq: mhat}
\hat{m}(x)=\mathbf{e}_1^T \left(\dfrac{1}{n}\mathbf{C}_x^T \mathbf{\Delta} \mathbf{C}_x+\mathbf{H}\right)^{-1}\dfrac{1}{n}\mathbf{C}_x^T\mathbf{\Delta} \mathbf{Y},
\end{equation}
where $\mathbf{C}_x=\begin{pmatrix} \mathbf{1} & \mathbf{C} \end{pmatrix}$, and $\mathbf{\Delta}$ is the diagonal weight matrix with $\{\Delta_i\}_{i=1}^n$ as entries. 

We set up the matrix $\mathbf{H}^*$ to accommodate different levels of roughness penalty for $\hat{\beta}_j$'s. As $\hat{\beta}_j$'s estimate the first order derivatives $m'_x$ along each basis direction $\phi_j$, variation in penalty is reasonable. There has been little work discussing multidimensional ridge penalization  applied to nonparametric functional regressions. Reiss et al.\ (2017 \cite{reiss2017}) established a real data example of signature verification using a scalar-on-function principal coordinate model, with a fixed-value ridge parameter. Seifert and Gasser (2000 \cite{seifert2000}) pointed out that, in multivariate local linear regression, it was unlikely to find a stable minimum among the whole space of nonnegative ridge matrices.    They mentioned a potential approach to iteratively search for optimal eigenvalues, given any set of eigenvectors, to minimize the mean squared error of the estimator, but they did not discuss this idea further. 

Here we develop a data-adaptive $\mathbf{H}^*$ that is amenable to theoretical work and allows a stable implementation of a multidimensional ridge penalty. Let $\mathbf{H}^*$ be in the class of matrices diagonalizable by $\mathbf{V} :\mathbf{H}^* \in \mathcal{R}(J, \mathbf{V}):=\{\mathbf{R}^{J \times J}: \mathbf{V}^T\mathbf{R}\mathbf{V} \text{ is diagonal}\}$, where $\mathbf{V}$ is the eigenvector matrix of the weighted sample covariance of scores $\langle X_i-x, \phi_j \rangle$, $1 \le j \le J$ (detailed discussion of $\mathbf{V}$ is included in Section \ref{sec: estimation}). Thus, $\mathbf{H}^*=\mathbf{V} \mathbf{\Lambda} \mathbf{V}^T$, $\mathbf{\Lambda}$ the diagonal matrix with entries $\lambda_j \ge 0$, $1 \le j \le J$.

The data-based matrix $\mathbf{V}$ carries out a change of basis along the eigenfunctions of the weighted covariance operators based on $x$: $\phi^*_j = \mathbf{V}_j^T\mathbf{\Phi}$ $\in \mathcal{L}^2(\mathcal{T})$, $1 \le j \le J$, $\mathbf{\Phi}$ as defined at the end of Section \ref{sec:fllr}. Consequently, $\boldsymbol{\beta}^*=\mathbf{V}^T \boldsymbol{\beta}$, the derivative functional $m'_x$ projected onto the new directions, can be estimated with a multidimensional roughness penalty to avoid overfitting. When all $\lambda_j$'s are equal, this becomes the special case of applying univariate ridge penalty $\lambda$ on squared norm of $\boldsymbol{\beta}$. Throughout later calculations such as Eq.~(\ref{eq: estbias2}) and (\ref{eq: var2}), the structure of $\mathbf{H}^*$ facilities an asymptotic analysis and, in applications, avoids the potential instability of multivariate cross validation when using a more general penalty matrix.

\section{Mean Squared Error (MSE) and Parameter Selection} \label{sec: estimation}
In this section, we estimate the mean squared error (MSE) of the FLLR-r estimator $\hat{m}(x)$ and explore the FLLR-r estimator's asymptotic behavior. The optimal ridge parameters are selected by minimizing the finite-sample estimated MSE by quadratic programming. 

\subsection{Estimated Bias and Variance}
We start with the conditional bias of $\hat{m}(x)$: 
\begin{align} \label{eq: bias}
\text{bias}\left(\hat{m}(x)\right) &=E[\hat{m}(x)|X_1, \ldots, X_n]-m(x) \nonumber\\
&=\mathbf{e}_1^T\left(\dfrac{1}{n} \mathbf{C}_x^T \mathbf{\Delta} \mathbf{C}_x+\mathbf{H}\right)^{-1}\dfrac{1}{n} \mathbf{C}_x^T\mathbf{\Delta}\left[B_1\left(x\right), B_2\left(x\right), \ldots, B_n\left(x\right)\right]^T,
\end{align}
where $B_i\left(x\right) = m'_x(X_i-x)+\dfrac{1}{2} m''_{r_i}\left(\left(X_i - x\right)^2\right)$ and $r_i = x + \rho_i \left(X_i - x\right)$ for some $\rho_i \in (0, 1)$, $i=1, \ldots, n$. We  estimate the MSE using an estimate of the truncated bias which is
\begin{align} \label{eq: estbias}
\text{bias}^{(J)} \left(\hat{m}(x)\right) &=\mathbf{e}_1^T\left(\dfrac{1}{n} \mathbf{C}_x^T \mathbf{\Delta} \mathbf{C}_x+\mathbf{H}\right)^{-1}\dfrac{1}{n} \mathbf{C}_x^T\mathbf{\Delta}\mathbf{C}\mathbf{m}'_{x,J}. 
\end{align}
In addition, the exact variance of $\hat{m}(x)$ is
\begin{align}
\text{Var} \left(\hat{m}\left(x\right)\right) 
&= n^{-2} \sigma_e^2 \mathbf{e}_1^T(n^{-1} \mathbf{C}_x^T \mathbf{\Delta} \mathbf{C}_x+\mathbf{H})^{-1}\mathbf{C}_x^T\mathbf{\Delta}^2 \mathbf{C}_x(n^{-1} \mathbf{C}_x^T \mathbf{\Delta} \mathbf{C}_x+\mathbf{H})^{-1}\mathbf{e}_1 \label{eq: var}
\end{align}
Based on Eq.~(\ref{eq: estbias}) and (\ref{eq: var}), we define the truncated MSE as $\text{MSE}_x\left(J, \mathbf{H}, h\right) = \left\{\text{bias}^{(J)}\left(\hat{m}(x)\right)\right\}^2 + \text{Var} (\hat{m}(x))$. We then reconstruct $\text{MSE}_x$ using the new orthogonal basis $\phi_1^*$, $\ldots \phi_J^*$.

\subsection{Reconstructed MSE and Ridge Penalty Optimization} \label{subsec: QP}
As stated in Section \ref{sec:ridge}, the columns of $\mathbf{V}$ for the basis change are the eigenvectors of the weighted sample covariance of the projected scores $\langle X_i, \phi_j \rangle$, $1 \le j \le J$. We define two weighted sample statistics of the projected scores: let $\hat{\mu}_j$ be the weighted average where $\hat{\mu}_j = \sum_i \Delta_i \langle X_i - x, \phi_j\rangle/\sum_i \Delta_i$, and similarly the weighted sample covariance is $\hat{\sigma}_{jk} = \sum_i \Delta_i \langle X_i - x, \phi_j\rangle \langle X_i - x, \phi_k\rangle/\sum_i \Delta_i - \hat{\mu}_j \hat{\mu}_k$ for scores on basis $\phi_j$ and $\phi_k$. Therefore, $\hat{\boldsymbol{\mu}} = \mathbf{C}^T\mathbf{\Delta}\mathbf{1}/\sum_i \Delta_i = (\hat{\mu}_1, \ldots, \hat{\mu}_k)^T$, and the $J \times J$ sample covariance matrix
$\mathbf{W}_{x, J}$ is
\begin{equation} \label{eq: wtcov}
\mathbf{W}_{x, J} =\left( \mathbf{C}^T-\hat{\boldsymbol{\mu}}\mathbf{1}^T\right)\mathbf{\Delta} \left( \mathbf{C}-\mathbf{1}\hat{\boldsymbol{\mu}}^T\right) \Big / \sum_i \Delta_i =\mathbf{C}^T\mathbf{\Delta}\mathbf{C}\Big /\sum_i \Delta_i-\hat{\boldsymbol{\mu}}\hat{\boldsymbol{\mu}}^T.
\end{equation}
Columns of the matrix $\mathbf{V}$ are the eigenvectors of $\mathbf{W}_{x, J}$.  Positive semi-definiteness of $\mathbf{W}_{x, J}$ is proved in the supplementary material.

With $a=(n^{-1} \sum_{i=1}^n \Delta_i)^{-1}$, we define
\begin{align}\label{eq: M}
    \tilde{\mathbf{M}} &= a^{-1}\mathbf{W}_{x, J} = \mathbf{V} \mathbf{\Lambda}^* \mathbf{V}^T \nonumber \\
    &=n^{-1}  \mathbf{C}^T \mathbf{\Delta} \mathbf{C} - a\left(n^{-1}  \mathbf{C}^T \mathbf{\Delta} \mathbf{1}\right) \left( n^{-1} \mathbf{1}^T \mathbf{\Delta} \mathbf{C}\right),
\end{align} 
where $\mathbf{\Lambda}^*$ is the diagonal matrix containing eigenvalues $\tilde{\gamma}_j > 0$, $\forall 1 \le j \le J$, of $\tilde{\mathbf{M}}$. By the asymptotic properties from Ba\'illo and Gran\'e (2009 \cite{baillo2009}), $a^{-1} = 1+o_p\left((n\psi_x(h))^{-1/2}\right)$. Then, $\mathbf{M}=\tilde{\mathbf{M}}+\mathbf{H}^*=\mathbf{V} \mathbf{D}\mathbf{V}^T$, where $\mathbf{D}=\mathbf{\Lambda}^*+\mathbf{\Lambda}$. Some other key terms denoted are:
\begin{itemize}
    \item $\mathbf{d}_1=\mathbf{V}^T\left(\dfrac{1}{n}\mathbf{C}^T\mathbf{\Delta}\mathbf{1}\right) = a^{-1}\hat{\boldsymbol{\mu}}^*$, where $\hat{\boldsymbol{\mu}}^* = \mathbf{V}^T \hat{\boldsymbol{\mu}}$ is the weighted sample average of scores $\langle X_i - x, \phi_j\rangle$ on the new basis $\phi_j^*$, $1 \le j \le J$;
    
    \item $d_2=n^{-1} \mathbf{1}^T\mathbf{\Delta} \mathbf{C} \mathbf{m}'_{x,J} = a^{-1} \langle \hat{\boldsymbol{\mu}}, \mathbf{m}'_{x,J}\rangle$;
    
    \item $\mathbf{d}_3=\mathbf{V}^T\left(\dfrac{1}{n}\mathbf{C}^T \mathbf{\Delta} \mathbf{C}\right)\mathbf{m}'_{x,J} = a^{-1} \mathbf{V}^T\left(\mathbf{W}_{x, J} + \hat{\boldsymbol{\mu}}\hat{\boldsymbol{\mu}}^T \right)\mathbf{m}'_{x,J}$.
\end{itemize}
   After some calculation, the bias from Eq.(\ref{eq: estbias}) is re-expressed as
\begin{align} \label{eq: estbias2}
\text{bias}^{(J)} (\hat{m}(x)) 
&= \langle \hat{\boldsymbol{\mu}}, \mathbf{m}'_{x,J}\rangle + a^{-1} \langle \hat{\boldsymbol{\mu}}, \mathbf{m}'_{x,J}\rangle \langle\hat{\boldsymbol{\mu}}^*, \mathbf{D}^{-1} \hat{\boldsymbol{\mu}}^*\rangle \nonumber\\
&\qquad - \hat{\boldsymbol{\mu}}^*{}^T\mathbf{D}^{-1}\mathbf{V}^T\left(\mathbf{W}_{x, J} + \hat{\boldsymbol{\mu}}\hat{\boldsymbol{\mu}}^T \right)\mathbf{m}'_{x,J} \nonumber\\
&=ad_2+a^2 d_2 \cdot \mathbf{d}_1^T\mathbf{D}^{-1}\mathbf{d}_1-a \cdot \mathbf{d}_1^T\mathbf{D}^{-1}\mathbf{d}_3,
\end{align}
and also the exact variance of $\hat{m}(x)$ in Eq.(\ref{eq: var}) is
\begin{align} \label{eq: var2}
\text{Var} (\hat{m}(x))
&=\left\|\sigma_e \mathbf{\Delta} \dfrac{1}{n} \left[a\mathbf{1} + \left(a^2 \mathbf{1}\mathbf{d}_1^T-a \mathbf{C}\mathbf{V}\right)\mathbf{D}^{-1}\mathbf{d}_1\right] \right\|^2
\end{align}

Detailed calculations are included in the supplementary materials.  Having $\mathbf{H}^*$ in the class of $\mathcal{R}(J, \mathbf{V})$ circumvents the complication of general matrix inversion, and transforms the problem of building $\text{MSE}_x\left(J, \mathbf{H}, h\right)$-optimal $\mathbf{H}^*$ into a quadratic programming problem with parameters $\lambda_1$, \ldots, $\lambda_J$, which are stored only in $\mathbf{D}$.

For brevity, we let 
\begin{itemize}
    \item $\mathbf{D}_1^*=\text{diag} \{\mathbf{d}_1\}$, $\mathbf{A}^*_1=(a^2d_2\mathbf{d}_1^T-a\mathbf{d}_3^T)\mathbf{D}^*_1$, $\mathbf{A}^*_2=\dfrac{1}{n} \sigma_e \mathbf{\Delta} (a^2\mathbf{1}\mathbf{d}_1^T-a\mathbf{C}\mathbf{V})\mathbf{D}^*_1$;
    
    \item $S_1=-ad_2$, $\mathbf{S}_2=-\dfrac{1}{n}a\sigma_e\mathbf{\Delta} \mathbf{1}$;
    
    \item $1/\tilde{\boldsymbol{\gamma}}=(1/\tilde{\gamma}_1, \ldots, 1/\tilde{\gamma}_J)$, and $\mathbf{b}=\left\{(\tilde{\gamma}_1+\lambda_1)^{-1}, \ldots, (\tilde{\gamma}_J+\lambda_J)^{-1}\right\}^T$.
\end{itemize}
With Eq.~(\ref{eq: estbias2}) and (\ref{eq: var2}), we search for optimal $\mathbf{b}^*$, and therefore optimal $\lambda_1^*, \ldots, \lambda_J^*$, by minimizing $\text{MSE}_x\left(J, \mathbf{H}, h\right)$: 
\begin{equation} \label{eq: leastsq}
\begin{split}
&\min_\mathbf{b} \text{MSE}_x\left(J, \mathbf{H}, h\right) 
= \min \left\{\|\mathbf{A}^*_1 \mathbf{b}-S_1\|^2+\|\mathbf{A}^*_2\mathbf{b}-\mathbf{S}_2\|^2\right\}, \\
&\qquad \qquad \qquad \qquad \text{  s.t.  } \mathbf{0} \le \mathbf{b} \le 1/\tilde{\gamma}.
\end{split}
\end{equation}
For $\mathbf{m}'_{x,J}$ in $d_2$ and $\mathbf{d}_3$, we use direct plug-in estimator $\hat{\boldsymbol{\beta}}^P$ from fitting Eq.~(\ref{eq: ridgerule}) where $\mathbf{H}^*$ is zero matrix, i.e., from the original FLLR rule, and $\hat \sigma_e$ is the standard error from FLLR fitting. 

\begin{remark} 
The data-adaptive structure of $\mathbf{H}^*$ enables the estimated MSE of $\hat{m}(x)$ to be written in a quadratic form in terms of the ridge parameters $\lambda_1, \ldots, \lambda_J$ for optimization, while general multivariate diagonal matrices would fail to do so. Based on Seifert and Gasser's (2000 \cite{seifert2000}) discussion of multivariate local polynomial regression, a more generic approach to find optimal eigenvalues of the general ridge matrix with other given sets of eigenvectors (unequal to $\mathbf{V}$) may be found iteratively, but at potentially high computational cost, while $\mathbf{H}^*=\mathbf{V}\mathbf{\Lambda}\mathbf{V}^T$ is not only empirically stable, but has desirable theoretical properties, which we discuss below.
\end{remark}

\subsection{Asymptotic Properties of FLLR-r}
Let $\mathcal{P}_{J}m'_x$ be the projection of the bounded linear functional $m'_x$ onto the subspace of $\mathcal{L}^2(\mathcal{T})$ spanned by $\phi_1, \ldots, \phi_J$ (also by $\phi^*_1, \ldots, \phi^*_J$), and $\mathcal{P}_{J^\perp} m'_x$ the projection onto the complementary subspace. We derive the asymptotic properties of $\hat{m}(x)$ as follows.

\begin{theorem}
Let Assumptions \ref{A1} - \ref{A3} hold. As $n \to \infty$, the conditional bias and variance of FLLR-r estimator $\hat{m}(x)$ are
\begin{enumerate}[i)]
    \item 
    $E\left(\hat{m}(x)|X_1, \ldots, X_n\right) - m(x) = O_P \left(\|\mathcal{P}_{J}m'_x\|h\right) + O_P \left(\|\mathcal{P}_{J^\perp} m'_x\| h\right) + O_P \left(h^2\right) + \\
    \kappa_J \cdot \left\{O_P\left(\|\mathcal{P}_{{J}^\perp} m'_x\| h^3\right) + O_P\left(h^4\right)\right\}$, where $\kappa_J = \max_{1\le j \le J} \dfrac{1}{\tilde{\gamma}_j+\lambda_j}$;
    \item
    $\text{Var}\left(\hat{m}(x)|X_1, \ldots, X_n\right) = O_P\left(\dfrac{1}{n\psi_x(h)}\right) +\kappa_J \cdot O_P\left(\dfrac{h^2}{n\psi_x(h)}\right) + \kappa_J^2 \cdot O_P\left(\dfrac{h^4}{n\psi_x(h)}\right)$.
\end{enumerate}
\end{theorem}

For the projected derivatives $\mathcal{P}_{J}m'_x$ and $\mathcal{P}_{J^\perp} m'_x$ in i), $m'_x = \mathcal{P}_{J}m'_x + \mathcal{P}_{J^\perp} m'_x$, and $\|m'_x\|^2 = \|\mathcal{P}_{J}m'_x\|^2 + \|\mathcal{P}_{J^\perp} m'_x\|^2$.  The sizes of both $\mathcal{P}_{J}m'_x$ and $\mathcal{P}_{J^\perp} m'_x$ are dependent on the magnitude of the derivative $m'_x$. As $J$ increases, $\|\mathcal{P}_{J}m'_x\| \to \|m'_x\|$ and $\|\mathcal{P}_{J^\perp} m'_x\| \to 0$.
The coefficient $\kappa_J$ is the minimum sum of a weighted covariance eigenvalue plus a corresponding ridge penalty. 
We here add an additional assumption \ref{A4} on $\kappa_J$ to discuss asymptotic behavior further. 

\begin{assump} \label{A4}
As $n \to \infty$ and $J=J(n) \to \infty$, $h^2/ \min_{1\le j \le J}\{\tilde{\gamma}_j+\lambda_j\} = O_P(1)$. Or equivalently, $h^2\kappa_J = O_P(1)$.
\end{assump}

Assumptions about the minimum eigenvalues of the score covariance matrices are not uncommon in local regression. See e.g., Reiss et al.(2017 \cite{reiss2017}), Ferraty and Nagy (2019 \cite{ferraty2019}). Here we are able to relax the restriction on the decay rates of the eigenvalues, as the ridge parameters can compensate for fast decreasing $\tilde{\gamma}_j$'s. With the additional Assumption \ref{A4}, bias and variance of $\hat{m}(x)$ are
\begin{corollary} \label{corollary}
With Assumptions \ref{A1}-\ref{A4},
\begin{enumerate}[i)]
    \item 
    $E\left(\hat{m}(x)|X_1, \ldots, X_n\right) - m(x) = O_P \left(\|\mathcal{P}_{J}m'_x\|h\right) + O_P \left(\|\mathcal{P}_{J^\perp} m'_x\| h\right) + O_P \left(h^2\right)$;
    \item
    $\text{Var}\left(\hat{m}(x)|X_1, \ldots, X_n\right) = O_P\left(\dfrac{1}{n\psi_x(h)}\right)$.
\end{enumerate}
\end{corollary}
Consequently, compared to FLLR estimator in Ferraty and Nagy (2019 \cite{ferraty2019}), the bias of $\hat{m}(x)$ has an additional term $O_P \left(\|\mathcal{P}_{J}m'_x\|h\right)$ from the ridge penalty, while the bound on the variance of $\hat{m}(x)$ is equivalent to FLLR's under \ref{A1} -- \ref{A4}.

\section{Simulation} \label{sec:sim}
We use simulated data to compare the performance of FLLR-r with the unpenalized local linear model FLLR, as well as the functional Nadaraya-Watson estimator (NW).  The functional Nadaraya-Watson estimator is a natural extension of its multivariate version, discussed in past work, e.g., Ferraty et al.\ (2007 \cite{ferraty2007}). NW estimates $m(x)$ as
\begin{equation*}
    \hat{m}^{NW}(x) = \dfrac{\sum_{i=1}^n Y_i K\left(\left\|X_i - x\right\|/h\right)}{\sum_{i=1}^n K\left(\left\|X_i - x\right\|/h\right)}
\end{equation*}
This section also discusses data-based selection of $h$ and $J$.

\subsection{Data Setup}
We use $201$ Fourier basis on $\mathcal{T} = [0, 1]$ for sample curve generation, where $\phi_{1}(t)=1, \phi_{2}(t)=\sqrt{2}\cos(2\pi t), \phi_{3}(t)=\sqrt{2}\sin(2\pi t), \ldots, \phi_{j}(t)=\sqrt{2} \cos(j\pi t)$ or $\sqrt{2} \sin \left(\left(j-1\right)\pi t\right)$ for $1< j \le 201$ according as $j$ is even or odd. With eigenvalues $\theta_j=1/j$, $X_i=\sum_{j=1}^{201} \sqrt{\theta_j}U_{ij}\phi_j$, where $U_{ij}$ is uniformly distributed i.i.d.\ scores on $[-\sqrt{3}, \sqrt{3}]$. The curves $X_i$ are observed on $51$ equispaced points, $t=0, 0.02, \ldots, 1$, on $\mathcal{T}=[0,1]$, with observation error $\xi_t \sim N(0, \sigma_t=0.2)$. Local linear pre-smoothing is applied with the direct plug-in bandwidth of Ruppert el al.\ (\cite{RSW1995}).
 
We follow the spirit of Ferraty and Nagy (2019 \cite{ferraty2019}) to design the regression $m$: $\mathcal{L}^2(\mathcal{T}) \to \mathbb{R}$ as a combination of linear and nonlinear models. Let 
\begin{align} \label{eq:sim}
    m(X_i) &= \left(1-a\right) \langle X_i , \sum_{j=1}^{30} \phi_j\rangle + a \sum_{j=1}^{20} \exp\left(-\langle X_i, \phi_j\rangle^2\right) \nonumber \\
    & = \left(1-a\right) \sum_{j=1}^{30} \sqrt{\theta_j}U_{ij} + a \sum_{j=1}^{20} \exp\left(-\theta_jU_{ij}^2\right),
\end{align}
where the sliding parameter $a \in [0, 1]$ varies the shape of $m$ between the linear regression (when $a=0$) and strongly nonlinear regression (when $a=1$). Random error $\epsilon_i \sim N(0, \sigma_e=0.5)$ is added to each observation: $Y_i = m(X_i) + \epsilon_i$.

\subsection{Selection of Tuning Parameters $J^*$, $h_r$, $h_d$} \label{sec: implement}
There are several global tuning parameters we must select for estimating $m(x)$: the optimal cut-off basis count $J^*$, regression bandwidth $h_r$ and derivative bandwidth $h_d$.   

As noted at the end of Section \ref{subsec: QP}, an estimated derivative vector at $x$, $\hat{\boldsymbol{\beta}}_x^P$, is necessary for the constructing ridge penalty. 
As an estimator, we use
\begin{equation*}
    \hat{\boldsymbol{\beta}}^P_x = \left[\mathbf{0} \quad \mathbf{I}\right]\left(\dfrac{1}{n}\mathbf{C}_x^T \mathbf{\Delta} \mathbf{C}_x\right)^{-1}\dfrac{1}{n}\mathbf{C}_x^T\mathbf{\Delta} \mathbf{Y},
\end{equation*}
where $[\mathbf{0} \quad \mathbf{I}]$ is a $J \times (J+1)$ matrix with a first column of 0's followed by an identity matrix. We obtain the preliminary $\hat{\boldsymbol{\beta}}^P_x$ from FLLR, i.e., without a ridge penalty, using the bandwidth $h_d$ discussed below.   Then, using $\hat{\boldsymbol{\beta}}^P_x$ we estimate $m(x)$ with a different bandwidth, $h_r$, by FLLR-r fitting, i.e., with a ridge penalty. Ferraty and Nagy (2019 \cite{ferraty2019}) mentioned that the asymptotic behavior of the estimated regression operator and its derivative are different, which is the motivation for using two distinct bandwidths $h_r$, $h_d$.

Nested leave-one-out cross-validation (LOOCV) is used for $J^*$ and $h_r$, but it is not suitable for $h_d$, as there is no direct way to measure the fitness of $\hat{\boldsymbol{\beta}}^P_x$.   Instead, we adopt wild bootstrapping of residuals to select $h_d$. The wild bootstrap was proposed by Wu (1986 \cite{wu1986}), and Ferraty et al.\ (2007 \cite{ferraty2007}) introduced it for bandwidth selection in nonparameteric functional regression.  Later, this method was applied to first-order functional derivative estimation by Ferraty and Nagy (2019 \cite{ferraty2019}). Also, Slaoui (2020 \cite{slaoui2020}) adopted the wild bootstrapping for bandwidth selection in recursive nonparametric functional regression.  

The tuning of the parameters $h_d$, $h_r$, and $J$ follows these steps:
\begin{enumerate}[i)]
    \item For each candidate cut-off basis $J$: use LOOCV to select the optimal FLLR bandwidth $h_{LL}$ which satisfies
    \begin{equation*}
        \min_h \dfrac{1}{n}\sum_{i=1}^n \left(Y_i - \hat{m}^{(-i)}_{LL}\left(X_i | h, J\right)\right)^2,
    \end{equation*}
    where $\hat{m}^{(-i)}_{LL}$ is the FLLR estimated regression operator at $X_i$, with $X_i$ removed in training.  In addition, denote the estimated derivative at $X_i$ using $h_{LL}\left(J\right)$ as $\tilde{\boldsymbol{\beta}}_{X_i} (h_{LL})$, which is estimated simultaneously with $\hat{m}^{(-i)}_{LL}$. Note that $h_{LL}$ is dependent on~$J$.
    
    \item 
    Define the residuals $\hat{\epsilon}_i = Y_i - \hat{m}^{(-i)}_{LL}\left(X_i | h_{LL}, J\right)$. Let the wild bootstrapped residuals be $\epsilon_i^b = \hat{\epsilon}_i\cdot v_i^b$, where $v_i^b$, $i=1, \ldots, n$ are i.i.d.\ random variables with $E(v_i^b) = 0$, and the next several moments equal to $1$. We use the most common choice, Mammen's two-point distribution (Mammen, 1993 \cite{mammen1993}):
    \begin{equation*}
        v_i^b = 
        \begin{cases}
        -\left(\sqrt{5}-1\right) /2, \text{  with probability } \left(\sqrt{5} + 1\right)/\left(2\sqrt{5}\right), \\
        \left(\sqrt{5}+1\right) /2,  \text{  with probability } \left(\sqrt{5} - 1\right)/\left(2\sqrt{5}\right).
        \end{cases}
    \end{equation*}
    In this case, $E(v_i^b) = 0$, $E\{(v_i^{b})^2\}= 1$, and $E\{(v_i^{b})^3\} = 1$, which ensure that the bootstrapped residuals $\epsilon_i^b$ have same first three moments as $\hat{\epsilon}_i$, $i=1, \ldots, n$ (see e.g.\ \cite{mackinnon2012}, \cite{mammen1993}). Other choices of $v_i$ include the Rademacher distribution (Davidson and Flachaire, 2008 \cite{davidson2008}) and Mammen's continuous distribution (Mammen, 1993 \cite{mammen1993}).

    \item Set $Y_i^{b} = \hat{m}^{(-i)}_{LL}\left(X_i | h_{LL}, J\right) + \epsilon_i^b$. For each of $b = 1, \ldots, B$ repetitions, estimate the derivative at $X_i$ with bandwidth $h$ as $\tilde{\boldsymbol{\beta}}_{X_i}^b (h)$ using the new set of data $\left(X_i, Y_i^{b}\right)$. Let $\hat{\boldsymbol{\beta}}_{X_i} (h) = \sum_{b=1}^B \tilde{\boldsymbol{\beta}}_{X_i}^b (h) / B$.
    
    \item
    Then, choose $h_d$ as the global bandwidth for the preliminary derivative estimation:
    \begin{equation} \label{eq: hdrule}
        h_d = \text{argmin}_h \dfrac{1}{n} \sum_{i=1}^n \left \| \tilde{\boldsymbol{\beta}}_{X_i} (h_{LL}) - \hat{\boldsymbol{\beta}}_{X_i} (h) \right\|^2. 
    \end{equation}
  Due to the difficulty of functional derivative estimation, Ferraty and Nagy (2019 \cite{ferraty2019}) designed the ad hoc bandwidth selector which minimizes the variation of the estimated derivative using $h_d$ from the one using $h_{LL}$ as in \eqref{eq: hdrule}, but doing this ignores the bias introduced by the latter. Future research can focus on developing a more systematic estimator for the functional derivatives.

    \item
    After the estimated derivative $\hat{\boldsymbol{\beta}}_x^P$ is calculated using bandwidth $h_d$, we plug $\hat{\boldsymbol{\beta}}_x^P$ into $d_2$ and $\mathbf{d_3}$ in Section \ref{subsec: QP} and search for the optimal $\lambda_j$'s in Eq.~\ref{eq: leastsq}. LOOCV is applied to select the global bandwidth $h_r$ for FLLR-r regression. In addition, since $h_d$, $h_r$ are all dependent on $J$, the optimal $J^*$ for FLLR and FLLR-r is determined through the nested LOOCV steps i) to v).

\end{enumerate}

\subsection{Model Performance Comparison} \label{sec:comparison}
We simulated $200$ Monte Carlo repetitions of model Eq.~(\ref{eq:sim}), each with $n_T=100$ training and $n_t=50$ test cases. To compare estimator performance at different levels of linearity of the regression function $m$, we implemented multiple models with $a = 0.3, 0.4, 0.5, 0.6, 0.7, 0.8$.  Larger $a$ implies stronger nonlinearity. 

The candidate cut-off $J$ values for LOOCV ranged from $1$ to $15$. In addition, for computational convenience, we translated each of the continuous bandwidths $h_{LL}$, $h_d$, $h_r$ to a discrete parameter $k_h$, which is the number of nearest neighbors of $x$. This technique was adopted from Ferraty et al.\ (2007 \cite{ferraty2007}),  and it was also applied in Ferraty and Nagy (2019 \cite{ferraty2019}). The maximum percentage of training cases that can selected as neighbors was set to $70\%$.

\begin{table}[h!]
    \centering
    \begin{tabular}{r|r|r|r|r|r|r}
    \hline
         & $a=0.3$ & $a=0.4$ & $a=0.5$ & $a=0.6$ & $a=0.7$ & $a=0.8$  \\
         \hline
     FLLR & 0.376 & 0.426 & 0.501 & 0.616 & 0.761 & 0.850\\
     \hline
     FLLR-r & \textbf{0.367} & \textbf{0.413} & \textbf{0.475} & \textbf{0.571} & \textbf{0.689} & \textbf{0.810}\\
     \hline
     NW & 0.544 & 0.572 & 0.616 & 0.682 & 0.761 & 0.858\\
     \hline
    \end{tabular}
    \caption{Averaged error ratios of prediction by FLLR, FLLR-r, and NW.  The optimal result for each $a$ is in bold. FLLR-r has the smallest error ratio in all scenarios.}
    \label{tab:sim_mse}
\end{table}

Table \ref{tab:sim_mse} records the averaged error ratios of prediction by each of the three methods at each level of nonlinearity, $a$. Error ratio (ER) is calculated by $\sum_{i=1}^{n_t} \left(Y_i - \hat{m}^*\left(X_i\right)\right)^2/$ $\sum_{i=1}^{n_t} \left(Y_i - \bar{Y}\right)^2$, where $\hat{m}^*\left(X_i\right)$ is estimated regression on $i$-th test case by each method, and $\bar{Y}$ is average of $Y_i$'s. ER is essentially the same as the widely used metric for regression, $(1-R^2)$.  As $a$ increases, the overall level of ER increases as well, but FLLR-r always achieves the best performance among the three methods.

\begin{figure}[h!] 
\centering
      \includegraphics[scale=0.525]{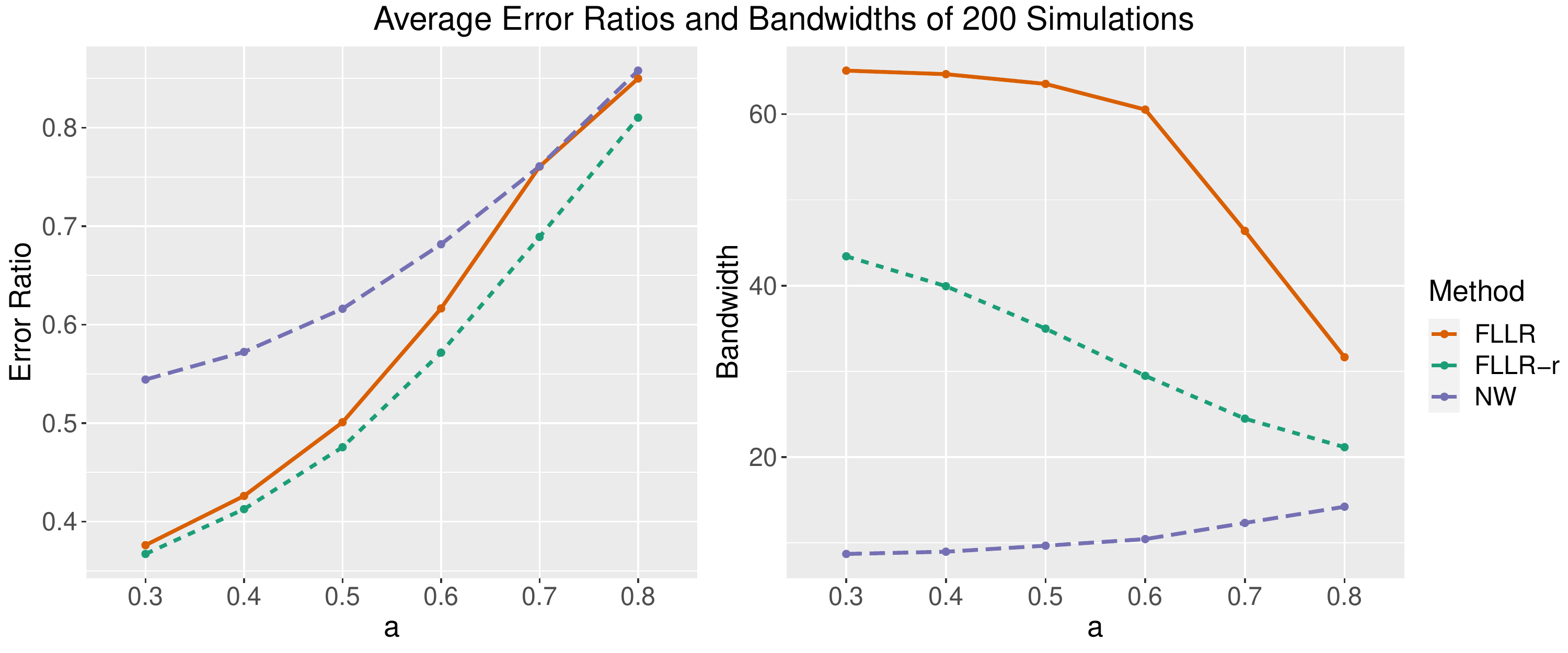}
  \caption{Plots of average error ratios (left) and bandwidths (right) by each method for $a$ from $0.3$ to $0.6$. FLLR-r achieves the lowest ER among the three methods, and uses  a  smaller bandwidth than FLLR.}
  \label{fig:sim_line}
\end{figure}

\begin{figure}[h!] 
\centering
      \includegraphics[scale=0.525]{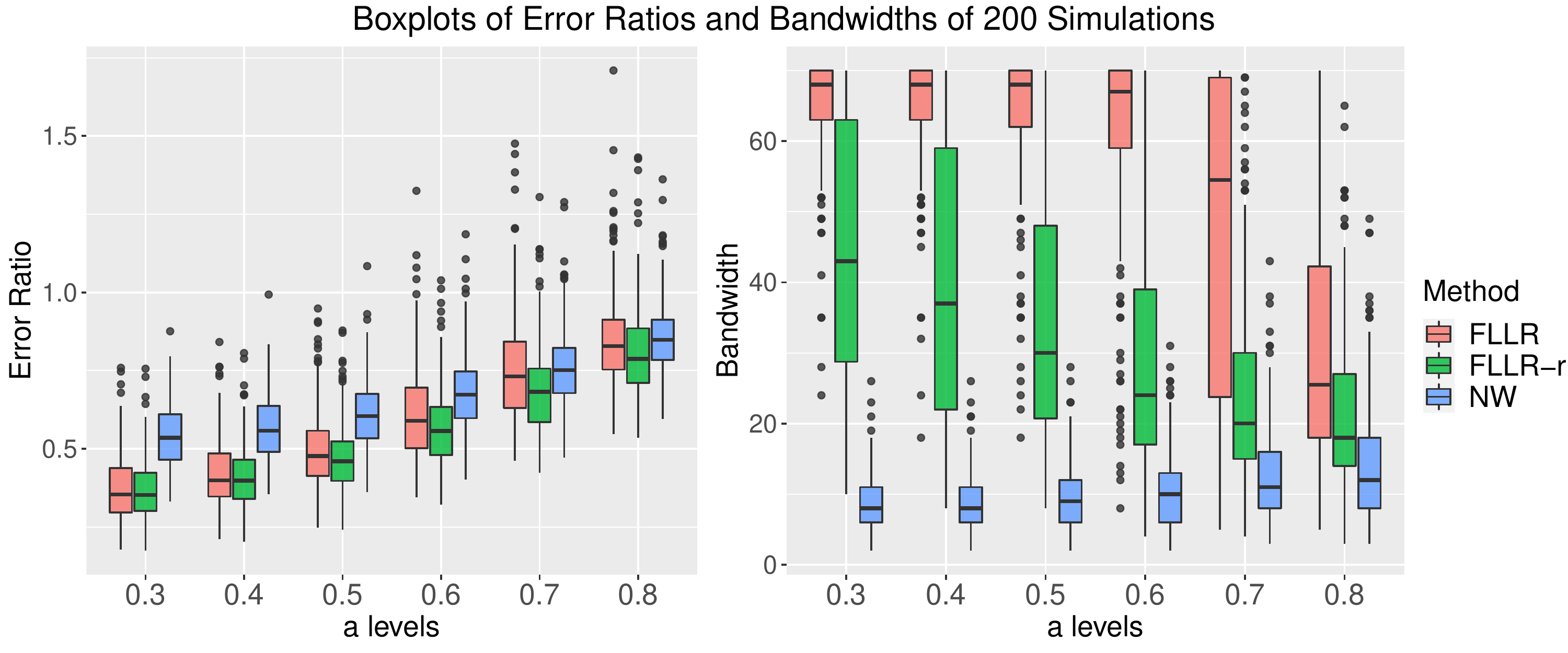}
  \caption{Boxplots of simulation error ratios (left) and cross-validated bandwidths (right) for FLLR, FLLR-r and NW at different levels. FLLR-r is advantageous in prediction especially at higher nonlinearity levels, and it needs smaller bandwidth for finite sample data.}
  \label{fig:sim_box}
\end{figure}

Fig.\ \ref{fig:sim_line} and \ref{fig:sim_box} summarize, at different $a$ levels, the performance of the three methods by error ratios and their selected bandwidth $k_{h_r}$ for the ridge estimator. According to the plots, FLLR and FLLR-r have very close prediction errors at lower levels of $a$. However, as the linearity of regression operator decreases with $a$ increasing, the performance of FLLR and FLLR-r diverge. At higher $a$, FLLR shows larger prediction errors and more outliers than FLLR-r, as seen in the left panel of Fig.~\ref{fig:sim_box}. Also, FLLR-r requires a smaller bandwidth at each level of $a$ in comparison with FLLR, as the right boxplot of Fig.~\ref{fig:sim_box} points out. The third quartile of the FLLR-r bandwidth among the $200$ simulations is below the first quartile of FLLR's for $a \le 0.6$.  The simulations show that FLLR-r is able to achieve smaller variation that FLLR. Such behavior is consistent with the ridge penalty in multivariate regression, which is known to reduce the variance of estimation while increasing its bias. 

In the supplementary material, we include results for derivative estimation by both FLLR and FLLR-r. Derivatives generated by FLLR-r tend to be flatter than FLLR.

\section{Two Real Data Examples} \label{sec:realdata}
\subsection{Particulate Matter (PM) Emission of Heavy Duty Trucks}
As our first example, we investigate the relationship between movement patterns of heavy duty trucks and particulate matter (PM) emissions.  We use the dataset in McLean et al.\ (2015 \cite{M2015})  originally extracted from the Coordinating Research Council E55/59 emissions inventory program documentary (Clark et al.\ 2007 \cite{C2007}). The dataset contains $108$ records of truck speeds in miles/hour over $90$ second intervals, and the logarithms of their PM emission in grams (log PM), captured by $70$ mm filters. We convert log PM back to the original PM weight by the exponential transformation. For each of the $200$ simulations, the dataset is randomly split into training and test cases by a ratio of $2:1$. Percentage of training cases considered for bandwidth selection is $50\%$.

\begin{figure}[h!] 
\centering
      \includegraphics[scale=0.525]{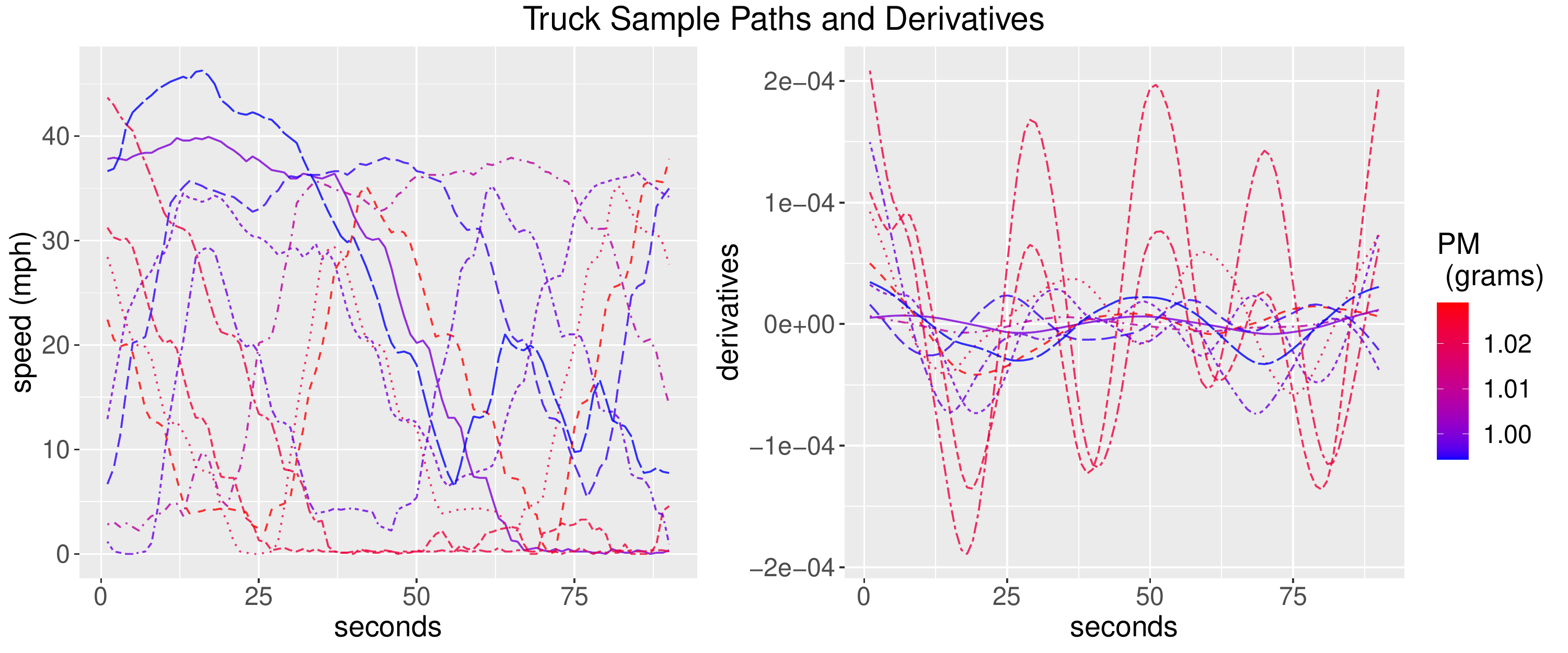}
  \caption{Plots of $10$ randomly sampled paths (left) and their corresponding estimated derivatives (right). Gradient color scale is used to represent the PM emission related to each sample. Derivatives on the right plot are calculated from estimated derivative scores $\hat{\boldsymbol{\beta}}_{X_i}^P$ (as in Section \ref{sec: implement}), $i = 1, \ldots, 10$, applied to the functional basis.}
  \label{fig:truck_sample}
\end{figure}

The left panel of Fig.~\ref{fig:truck_sample} shows $10$ randomly sampled paths of truck speed, where the gradient color scale corresponds to PM emissions in grams. The right panel includes the estimated derivative functions by FLLR-r for each case, calculated from scores $\hat{\boldsymbol{\beta}}_{X_i}^P$, $i = 1, \ldots, 10$, applied to the functional basis.  As the derivatives vary substantially across different records, we can safely infer that the regression function mapping truck movement patterns to PM emissions is nonlinear.

\begin{table}[h!]
    \centering
    \begin{tabular}{r|rrrr}
    \hline
         & FLLR & FLLR-r & NW & FLM \\
         \hline
       Error Ratio (ER) & 0.715 & \textbf{0.652} & 0.771 & 0.862\\
       \hline
       Mean $k_h$ & 30.4 & 20.3 & 6.5 & N/A\\
       \hline
    \end{tabular}
    \caption{Averaged error ratios of four models for $200$ repetitions and average num of nearest neighbors $k_h$ used by each method are also included.}
    \label{tab:truck_table}
\end{table}

\begin{figure}[h!] 
\centering
      \includegraphics[scale=0.525]{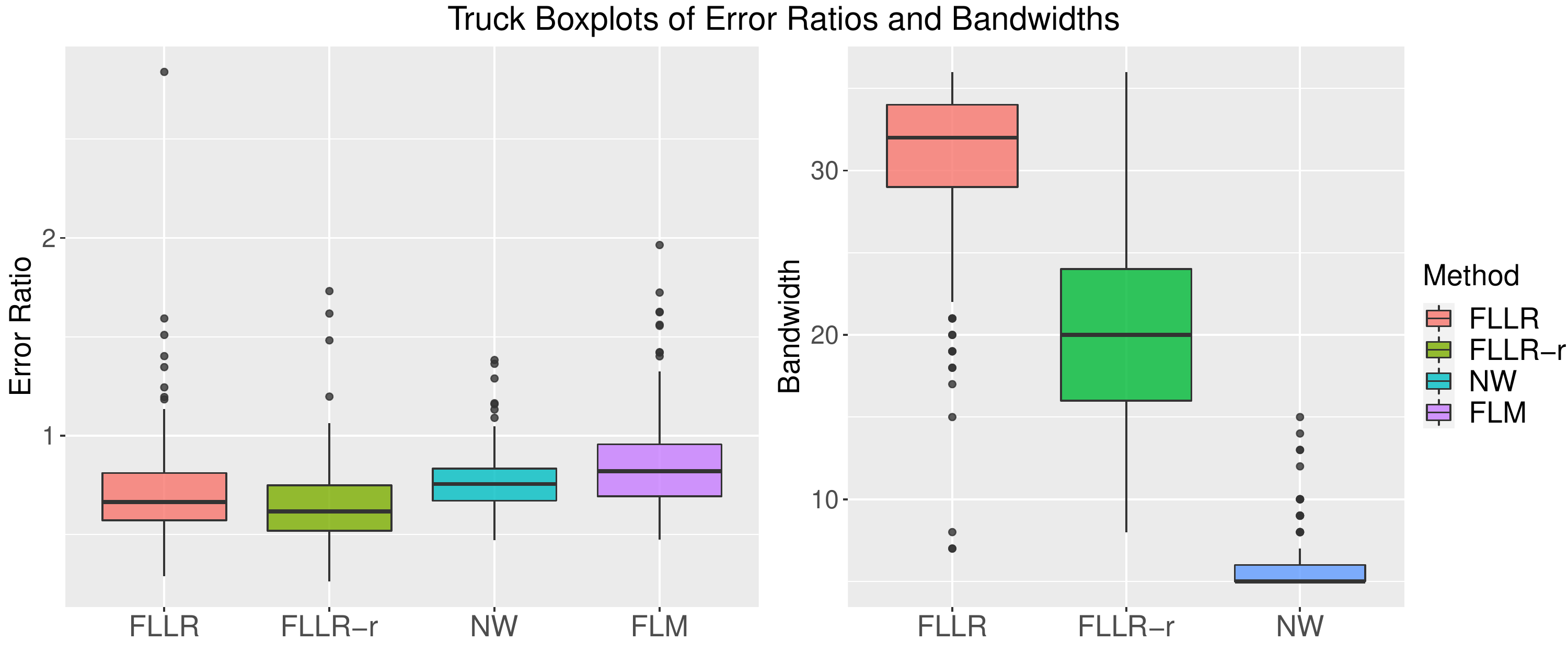}
  \caption{Boxplots of error ratios and bandwidths of the three methods for estimated regression of PM emission on truck speed.}
  \label{fig:truck_box}
\end{figure}

Table \ref{tab:truck_table} shows the averaged error ratios from $200$ repetitions by each estimator, as well as the  number of nearest neighbors selected by cross validation. 
An additional estimator, FLM,  the functional linear model with scalar ridge penalty in the R package `fda.usc'(\cite{fda_usc}), is included.   Cross validation is used for its ridge parameter tuning.
We see that FLLR-r is able to achieve the lowest error ratio and uses fewer nearest neighbors than \mbox{FLLR}. Average cut-off $J^*$ for FLLR and FLLR-r are $9.98$ and $11$ respectively. Boxplots in Fig.\ \ref{fig:truck_box}  show the advantages of FLLR-r in prediction accuracy and bandwidth choice. 

\subsection{Oil Content in Cargill Corn Samples}
The second example uses a data set of $80$ corn specimens measured with different NIR spectrometers at wavelengths $1100$--$2498$nm at $2$nm intervals. We choose instrument mp5 for analysis here. Oil content in percentage of total corn kernel weight is also recorded. We use FLLR, FLLR-r and NW to examine the regression mapping corn NIR data to oil content. The original data set can be accessed online at \url{https://eigenvector.com/resources/data-sets}. Again, the dataset is randomly split into training and test cases by a ratio of $2:1$ during each of the 200 simulations. Percentage of training cases considered for bandwidth selection is $50\%$.  As in the previous truck example, Fig.~\ref{fig:corn_sample} shows $10$ randomly sampled NIR paths, with derivatives $m'_{X_i}$ estimated by FLLR-r. 

\begin{figure}[h!] 
\centering
      \includegraphics[scale=0.525]{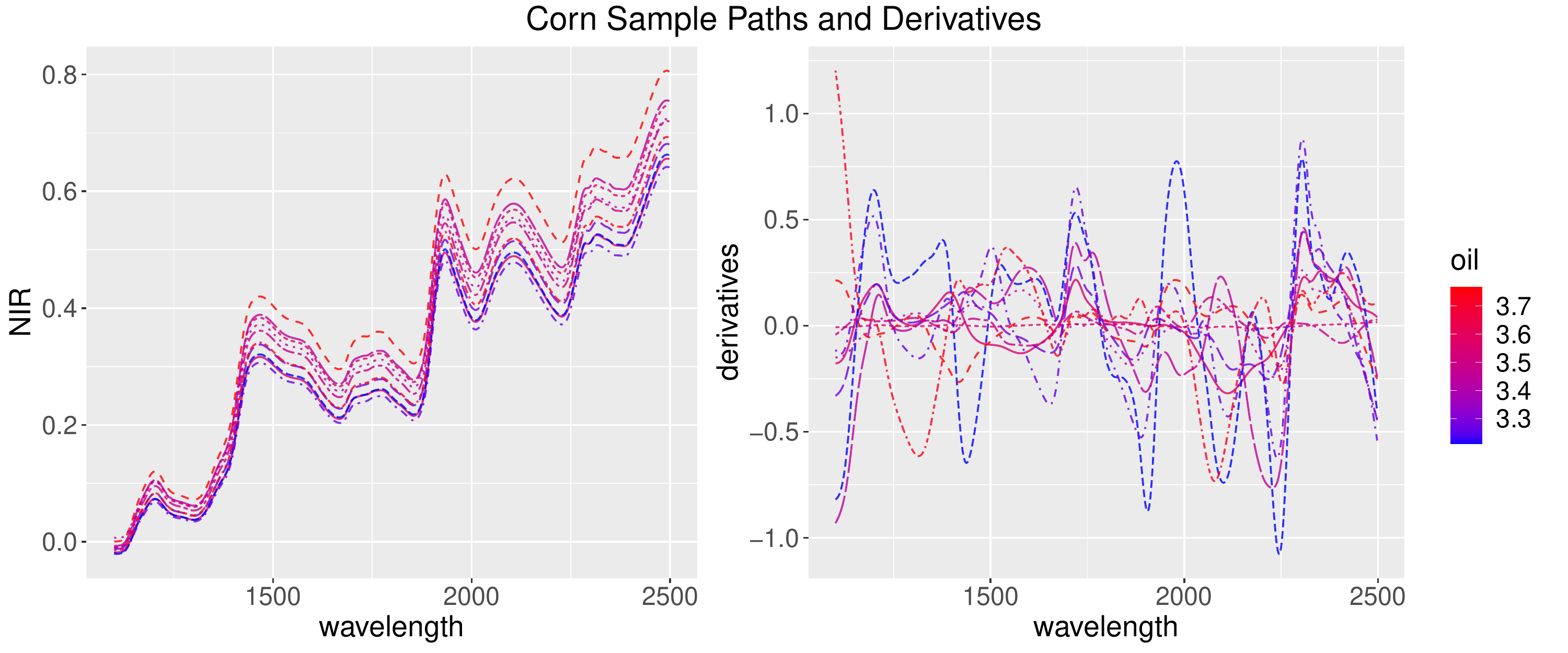}
  \caption{Plots of $10$ randomly selected corn samples with NIR paths (left) and their corresponding estimated derivatives $m'_{X_i}$ (right). A gradient color scale represents the oil content of each sample.}
  \label{fig:corn_sample}
\end{figure}

\begin{table}[h!]
    \centering
    \begin{tabular}{r|rrrr}
    \hline
         & FLLR & FLLR-r & NW & FLM\\
         \hline
       Error Ratio (ER) & 0.480 & 0.421 & 1.017 & \textbf{0.378}\\
       \hline
       Mean $k_h$ & 24.4 & 19.8 & 12.6 & N/A\\
       \hline
    \end{tabular}
    \caption{Averaged error ratios of four models on corn NIR data, for $200$ repetitions. Average numbers of nearest neighbors $k_h$ used by each method are also included.}
    \label{tab:corn_table}
\end{table}

\begin{figure}[h!] 
\centering
      \includegraphics[scale=0.525]{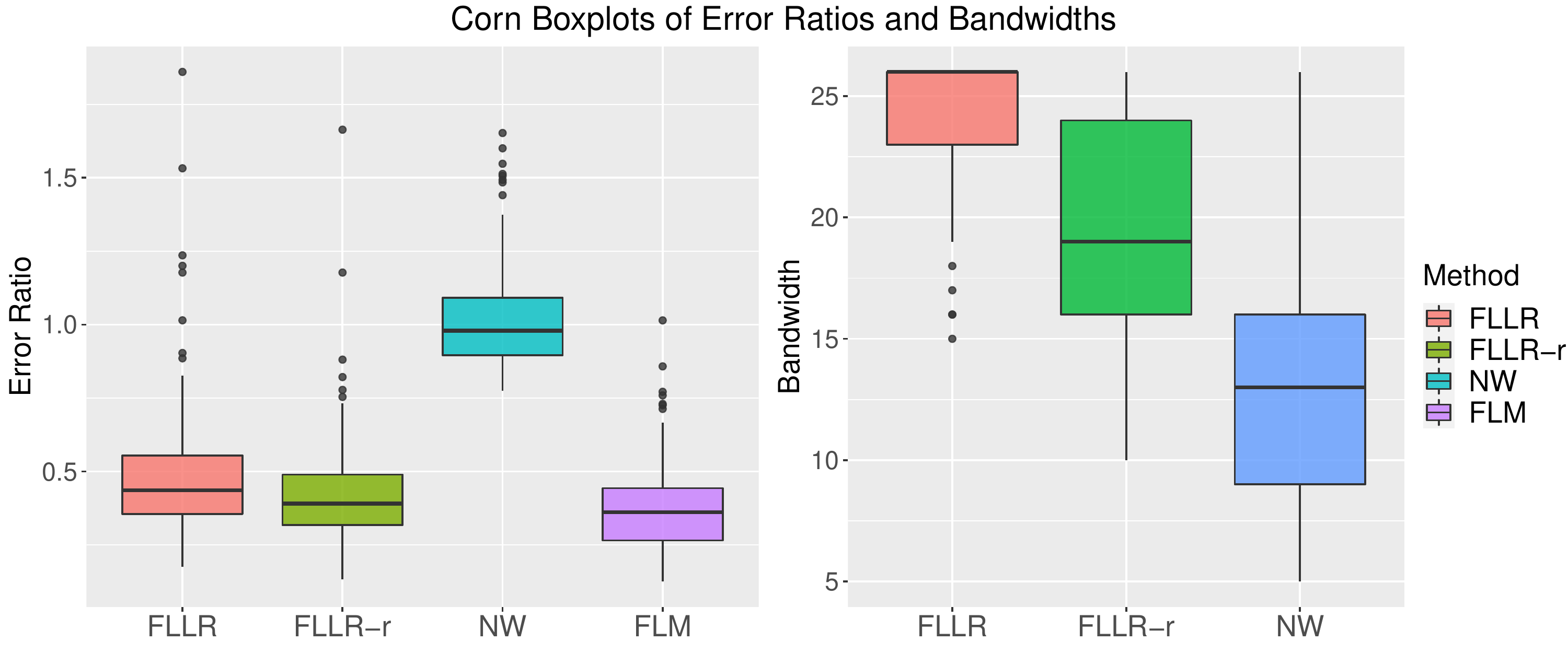}
  \caption{Boxplots of error ratios and bandwidths of the three methods for estimated regression of corn oil content on NIR.}
  \label{fig:corn_box}
\end{figure}

Table \ref{tab:corn_table} displays the averaged error ratios from $200$ repetitions by each estimator, as well as the optimal number of nearest neighbors they selected by cross validation. Fig.~\ref{fig:corn_box} displays the repetition results in boxplots. Due to the relatively higher linearity level compared to the previous example, we see that FLM performs best in this example. However, FLLR-r still generates lowest error ratios for prediction  among the three nonparametric methods, and FLLR-r uses a smaller bandwidth than FLLR.

\section{Discussion}
In this work, we extend multidimensional ridge penalization to functional regression and propose a specific type of ridge matrix that adapts to the weighted covariance of the sample scores.   We discuss in detail our tuning parameter selector, which is designed to minimize the mean squared error of predictions. Both theoretical results and data analysis show the advantages of this model (FLLR-r), including higher prediction accuracy, especially in regression with higher degree of nonlinearity and a reduction of variance. 

Estimation of functional derivatives is another important yet challenging topic in functional data analysis. Fan and Zhang (2000 \cite{fan2000two}) discussed estimating derivatives in functional linear models, and  M{\"u}ller et al.\ (2010 \cite{muller2010additive}) covered derivative estimation in functional additive models. However, there is relatively little work on nonparametric estimation of functional derivatives.
In ongoing work, we are developing further our methodology for functional derivative estimation. Potential directions include a efficient bandwidth selection for estimating derivatives and higher order functional local polynomial models for a more accurate approximation of the first-order differential operators.  We anticipate that a well-developed nonparametric derivative estimator can be applied to improve the FLLR-r model for better predictions.

\bibliographystyle{apalike}
\bibliography{references} 

\begin{thebibliography}{}

\bibitem[Ba{\'\i}llo and Gran{\'e}, 2009]{baillo2009}
Ba{\'\i}llo, A. and Gran{\'e}, A. (2009).
\newblock Local linear regression for functional predictor and scalar response.
\newblock {\em Journal of Multivariate Analysis}, 100(1):102--111.

\bibitem[Bande et~al., 2020]{fda_usc}
Bande, M.~F., de~la Fuente, M.~O., Galeano, P., Nieto, A., Garcia-Portugues,
  E., and de~la Fuente, M. M.~O. (2020).
\newblock Package ‘fda. usc’.
\newblock {\em CRAN Repository}.

\bibitem[Barrientos-Marin et~al., 2010]{barrientos2010locally}
Barrientos-Marin, J., Ferraty, F., and Vieu, P. (2010).
\newblock Locally modelled regression and functional data.
\newblock {\em Journal of Nonparametric Statistics}, 22(5):617--632.

\bibitem[Berlinet et~al., 2011]{berlinet2011}
Berlinet, A., Elamine, A., and Mas, A. (2011).
\newblock Local linear regression for functional data.
\newblock {\em Annals of the Institute of Statistical Mathematics},
  63(5):1047--1075.

\bibitem[Boj et~al., 2010]{boj2010distance}
Boj, E., Delicado, P., and Fortiana, J. (2010).
\newblock Distance-based local linear regression for functional predictors.
\newblock {\em Computational Statistics \& Data Analysis}, 54(2):429--437.

\bibitem[Cai et~al., 2006]{cai2006}
Cai, T.~T., Hall, P., et~al. (2006).
\newblock Prediction in functional linear regression.
\newblock {\em The Annals of Statistics}, 34(5):2159--2179.

\bibitem[Clark et~al., 2007]{C2007}
Clark, N.~N., Gautam, M., Wayne, W.~S., Lyons, D.~W., Thompson, G., and
  Zielinska, B. (2007).
\newblock Heavy-duty vehicle chassis dynamometer testing for emissions
  inventory, air quality modeling, source apportionment and air toxics
  emissions inventory.
\newblock {\em Coordinating Research Council, incorporated}.

\bibitem[Davidson and Flachaire, 2008]{davidson2008}
Davidson, R. and Flachaire, E. (2008).
\newblock The wild bootstrap, tamed at last.
\newblock {\em Journal of Econometrics}, 146(1):162--169.

\bibitem[Fan, 1992]{fan1992design}
Fan, J. (1992).
\newblock Design-adaptive nonparametric regression.
\newblock {\em Journal of the American statistical Association},
  87(420):998--1004.

\bibitem[Fan and Zhang, 2000]{fan2000two}
Fan, J. and Zhang, J.-T. (2000).
\newblock Two-step estimation of functional linear models with applications to
  longitudinal data.
\newblock {\em Journal of the Royal Statistical Society: Series B (Statistical
  Methodology)}, 62(2):303--322.

\bibitem[Ferraty et~al., 2007]{ferraty2007}
Ferraty, F., Mas, A., and Vieu, P. (2007).
\newblock Nonparametric regression on functional data: inference and practical
  aspects.
\newblock {\em Australian \& New Zealand Journal of Statistics},
  49(3):267--286.

\bibitem[Ferraty and Nagy, 2019]{ferraty2019}
Ferraty, F. and Nagy, S. (2019).
\newblock Scalar-on-function local linear regression and beyond.
\newblock {\em arXiv preprint arXiv:1907.08074}.

\bibitem[Ferraty and Vieu, 2006]{ferraty2006nonparametric}
Ferraty, F. and Vieu, P. (2006).
\newblock {\em Nonparametric functional data analysis: theory and practice}.
\newblock Springer Science \& Business Media.

\bibitem[MacKinnon, 2012]{mackinnon2012}
MacKinnon, J.~G. (2012).
\newblock Inference based on the wild bootstrap.
\newblock In {\em Seminar presentation given to Carleton University in
  September}.

\bibitem[Mammen, 1993]{mammen1993}
Mammen, E. (1993).
\newblock Bootstrap and wild bootstrap for high dimensional linear models.
\newblock {\em The annals of statistics}, pages 255--285.

\bibitem[McLean et~al., 2015]{M2015}
McLean, M.~W., Hooker, G., and Ruppert, D. (2015).
\newblock Restricted likelihood ratio tests for linearity in scalar-on-function
  regression.
\newblock {\em Statistics and Computing}, 25(5):997--1008.

\bibitem[M{\"u}ller and Yao, 2010]{muller2010additive}
M{\"u}ller, H.-G. and Yao, F. (2010).
\newblock Additive modelling of functional gradients.
\newblock {\em Biometrika}, 97(4):791--805.

\bibitem[Reiss et~al., 2017]{reiss2017}
Reiss, P.~T., Miller, D.~L., Wu, P.-S., and Hua, W.-Y. (2017).
\newblock Penalized nonparametric scalar-on-function regression via principal
  coordinates.
\newblock {\em Journal of Computational and Graphical Statistics},
  26(3):569--578.

\bibitem[Reiss and Ogden, 2007]{reiss2007}
Reiss, P.~T. and Ogden, R.~T. (2007).
\newblock Functional principal component regression and functional partial
  least squares.
\newblock {\em Journal of the American Statistical Association},
  102(479):984--996.

\bibitem[Ruppert et~al., 1995]{RSW1995}
Ruppert, D., Sheather, S.~J., and Wand, M.~P. (1995).
\newblock An effective bandwidth selector for local least squares regression.
\newblock {\em Journal of the American Statistical Association},
  90(432):1257--1270.

\bibitem[Ruppert and Wand, 1994]{ruppert1994}
Ruppert, D. and Wand, M.~P. (1994).
\newblock Multivariate locally weighted least squares regression.
\newblock {\em The annals of statistics}, pages 1346--1370.

\bibitem[Seifert and Gasser, 2000]{seifert2000}
Seifert, B. and Gasser, T. (2000).
\newblock Data adaptive ridging in local polynomial regression.
\newblock {\em Journal of Computational and Graphical Statistics},
  9(2):338--360.

\bibitem[Slaoui, 2020]{slaoui2020}
Slaoui, Y. (2020).
\newblock Recursive nonparametric regression estimation for independent
  functional data.
\newblock {\em Statistica Sinica}, 30(1):417--37.

\bibitem[Wu et~al., 1986]{wu1986}
Wu, C.-F.~J. et~al. (1986).
\newblock Jackknife, bootstrap and other resampling methods in regression
  analysis.
\newblock {\em the Annals of Statistics}, 14(4):1261--1295.

\bibitem[Zeidler, 1995]{zeidler1995}
Zeidler, E. (1995).
\newblock {\em Applied functional analysis: main principles and their
  applications}, volume 109.
\newblock Springer Science \& Business Media.

\end{thebibliography}


\begin{thebibliography}{}

\bibitem[Ba{\'\i}llo and Gran{\'e}, 2009]{baillo2009}
Ba{\'\i}llo, A. and Gran{\'e}, A. (2009).
\newblock Local linear regression for functional predictor and scalar response.
\newblock {\em Journal of Multivariate Analysis}, 100(1):102--111.

\end{thebibliography}

\end{document}


\maketitle

\section{Additional Theoretical Results}
\subsection{Positive Semi-definiteness of $\mathbf{W}_{x, J}$}
It can be proved that the matrix $(\sum \Delta_i) \Delta-\Delta \mathbf{1} \mathbf{1}^T \Delta$ is positive semi-definite: $\forall \mathbf{u}=(u_1, \ldots, u_n)^T \in \mathbb{R}^n$,
\begin{align} \label{supeq:pd}
\mathbf{u}^T(\sum \Delta_i) \Delta \mathbf{u}-\mathbf{u}^T \Delta \mathbf{1} \mathbf{1}^T \Delta \mathbf{u} &=\displaystyle \sum_{i=1}^n \Delta_i \sum_{t=1}^n \Delta_t u_t^2-(\sum_{t=1}^n \Delta_t u_t)^2 \nonumber \\
&=\sum_{i=1}^n \sum_{t=1}^n \Delta_i \Delta_t u_t^2- \sum_{i=1}^n \sum_{t=1}^n \Delta_i \Delta_t u_i u_t \nonumber\\
&=\sum_{i=1}^n \sum_{t=1}^n \Delta_i \Delta_t u_t (u_t-u_i) \nonumber \\
\end{align}
When $t=i$, $\Delta_i \Delta_t u_t (u_t-u_i)=0$;
when $t \ne i$, $\Delta_i \Delta_t u_t (u_t-u_i) + \Delta_t \Delta_i u_i (u_i-u_t)=\Delta_i \Delta_t (u_t-u_i)^2 \ge 0$. Thus, Eq.\ref{supeq:pd} $\ge 0$, i.e. $(\sum \Delta_i) \Delta-\Delta \mathbf{1} \mathbf{1}^T \Delta$ is positive semi-definite, and so is $\mathbf{W}_{x, J}$.

\subsection{Asymptotic Bound of FLLR-r Estimator Bias}
We decompose the inverse matrix $\left(\dfrac{1}{n} \mathbf{C}_x^T \mathbf{\Delta} \mathbf{C}_x+\mathbf{H}\right)^{-1}$ to rewrite bias and variance. Let $\left(\dfrac{1}{n} \mathbf{C}_x^T \mathbf{\Delta} \mathbf{C}_x+\mathbf{H}\right)^{-1}=\begin{pmatrix}A_{1,1} & \mathbf{A}_{1,2} \\ \mathbf{A}_{1,2}^T & \mathbf{A}_{2,2} \end{pmatrix}$, where $A_{1,1}$, $\mathbf{A}_{1,2}$, $\mathbf{A}_{2,2}$ are respectively the block matrices of size $1 \times 1$, $1 \times J$ and $J \times J$. 

Proof of the asymptotic bounds of $\hat{m}(x)$ needs the following lemma
\begin{lemma} \label{suplemma:1}
For i.i.d.\ samples $X_i$, $i=1, \ldots, n$, and $n \to \infty$, 
\begin{enumerate}[i)]
    \item $\left\|\dfrac{1}{n} \sum_i \Delta_i \mathbf{c}_i\right\| = Op\left(h\right)$;
    \item $\left\|\dfrac{1}{n}\mathbf{C}_p^T \mathbf{\Delta}^2\mathbf{1}\right\| = Op\left(h\dfrac{1}{\psi_x\left(h\right)}\right)$;
    \item $\left\|\dfrac{1}{n}\mathbf{C}_p^T \mathbf{\Delta}^2\mathbf{C}_p\right\| = Op\left(h^2\dfrac{1}{\psi_x\left(h\right)}\right)$.
\end{enumerate}

\end{lemma}

\begin{proof}
\begin{enumerate}[i)]
    \item 
    The proof directly uses Markov Inequality: $E\left\|\dfrac{1}{n}\sum_i \Delta_i \mathbf{c}_i\right\| \le E \left\|\Delta_i \mathbf{c}_i\right\| \le h$, so $\forall \epsilon$, $\exists M_{\epsilon} = 1/\epsilon$ s.t.\  $P\left(\left\|\dfrac{1}{n} \sum_i \Delta_i \mathbf{c}_i\right\|/h\right) 1/\epsilon) \le \epsilon$. So $\left\|\dfrac{1}{n} \sum_i \Delta_i \mathbf{c}_i\right\|/h = Op(1)$.
    \item 
    $E\left\|\dfrac{1}{n}\sum_i \Delta_i^2 \mathbf{c}_i\right\| \le E \left\|\Delta_i^2 \mathbf{c}_i\right\| \le E \left(\Delta_i^2\right)h$. $E \left(\Delta_i^2\right) = \dfrac{EK^2\left(\left\|X_i-x\right\|/h\right)}{\left(EK\left(\left\|X_i-x\right\|/h\right)\right)^2} \le \dfrac{C_k^2}{c_k^2\psi_x\left(h\right)}$. Thus, by Markov Inequality again, $\left\|\dfrac{1}{n}\mathbf{C}_p^T \mathbf{\Delta}^2\mathbf{1}\right\| = Op\left(h\dfrac{1}{\psi_x\left(h\right)}\right)$. 
    \item $\left\|\dfrac{1}{n}\mathbf{C}_p^T \mathbf{\Delta}^2\mathbf{C}_p\right\| = \left\|\dfrac{1}{n} \sum_i \mathbf{c}_i \Delta_i^2 \mathbf{c}_i^T\right\|\le \dfrac{1}{n}\sum_i \left\|\mathbf{c}_i \Delta_i^2 \mathbf{c}_i^T\right\| = \dfrac{1}{n}\sum_i \Delta_i^2\left\|\mathbf{c}_i\right\|^2$. Then with Markov Inequality as in i) and ii), $\left\|\dfrac{1}{n}\mathbf{C}_p^T \mathbf{\Delta}^2\mathbf{C}_p\right\| =Op\left(h^2\dfrac{1}{\psi_x\left(h\right)}\right)$.

\end{enumerate}

\end{proof}

First, bias of $\hat{m}(x)$ is calculated by
\begin{equation} \label{supeq:1}
    \begin{split}
    E[\hat{m}(x)|X_1, \ldots, X_n] &= \mathbf{e}_1^T\left(\dfrac{1}{n} \mathbf{C}_x^T \mathbf{\Delta} \mathbf{C}_x+\mathbf{H}\right)^{-1}\dfrac{1}{n}\mathbf{C}_x^T\mathbf{\Delta} \begin{pmatrix} m(x) + m'_x(X_1-x)+\dfrac{1}{2} m''_{r_1}\left(\left(X_1 - x\right)^2\right)\\\vdots\\m(x)+m'_x(X_n-x)+\dfrac{1}{2}m''_{r_n}\left(\left(X_n - x\right)^2\right)\end{pmatrix} \\
    &= \mathbf{e}_1^T\left(\dfrac{1}{n} \mathbf{C}_x^T \mathbf{\Delta} \mathbf{C}_x+\mathbf{H}\right)^{-1}\left(\dfrac{1}{n} \mathbf{C}_x^T\mathbf{\Delta} \mathbf{C}_x + \mathbf{H}\right) \left[m(x), \mathbf{m'_{x, J}}^T\right]^T \\
    &- \mathbf{e}_1^T\left(\dfrac{1}{n} \mathbf{C}_x^T \mathbf{\Delta} \mathbf{C}_x+\mathbf{H}\right)^{-1} \mathbf{H}\left[m(x), \mathbf{m'_{x, J}}^T\right]^T \\
    &+ \mathbf{e}_1^T\left(\dfrac{1}{n} \mathbf{C}_x^T \mathbf{\Delta} \mathbf{C}_x+\mathbf{H}\right)^{-1}\dfrac{1}{n} \mathbf{C}_x^T\mathbf{\Delta} \begin{pmatrix} \mathcal{P_{{S_J}^\perp}} m'_x (X_1 -x) \\ \vdots \\ \mathcal{P_{{S_J}^\perp}} m'_x (X_n -x) \end{pmatrix}\\
    &+ \mathbf{e}_1^T\left(\dfrac{1}{n} \mathbf{C}_x^T \mathbf{\Delta} \mathbf{C}_x+\mathbf{H}\right)^{-1}\dfrac{1}{n} \mathbf{C}_x^T\mathbf{\Delta} \begin{pmatrix} \dfrac{1}{2}m''_{r_1}\left(\left(X_1 - x\right)^2\right) \\ \vdots \\ \dfrac{1}{2}m''_{r_n}\left(\left(X_n - x\right)^2\right) \end{pmatrix}.
    \end{split}
\end{equation}

Apparently second line of Eq.(\ref{supeq:1}) is $m(x)$. So bias of $\hat{m}(x)$ is the sum of the last three lines of Eq.(\ref{supeq:1}), which we denote respectively as $T_0$, $T_1$ and $T_2$. 

Asymptotic bound of $T_0$ is
\begin{align}
    T_0 &= - \mathbf{e}_1^T\left(\dfrac{1}{n} \mathbf{C}_x^T \mathbf{\Delta} \mathbf{C}_x+\mathbf{H}\right)^{-1} \mathbf{H}\left[m(x), {\mathbf{m}'_{x, J}}^T\right]^T\nonumber \\
    &= - \left[ A_{11} \quad \mathbf{A}_{12}\right]\begin{bmatrix} 0 \\\mathbf{H}^*\mathbf{m}'_{x,J}\end{bmatrix} = - \mathbf{A}_{12} \mathbf{H}^* \mathbf{m}'_{x,J} \nonumber\\
    &= a\dfrac{1}{n} \mathbf{1}^T \mathbf{\Delta} \mathbf{C_p} \mathbf{V} \mathbf{D}^{-1}\mathbf{V}^T \mathbf{H}^*\mathbf{m}'_{x,J} \nonumber\\
    &= a\dfrac{1}{n} \mathbf{1}^T \mathbf{\Delta} \mathbf{C_p} \mathbf{V} \mathbf{D}^{-1} \mathbf{\Lambda}\mathbf{V}^T \mathbf{m}'_{x,J} \nonumber\\
    & \le a \left\|\dfrac{1}{n}\mathbf{1}^T \mathbf{\Delta} \mathbf{C_p} \mathbf{V}\right\| \cdot \left\|\mathbf{D}^{-1} \mathbf{\Lambda}\mathbf{V}^T \mathbf{m}'_{x,J}\right\| \nonumber\\
    & \le a \left\|\dfrac{1}{n}\sum_i \Delta_i \mathbf{c_i}\right\| \cdot \max_{1 \le j \le J}\dfrac{\lambda_j}{\tilde{\gamma}_j+\lambda_j} \left\|\mathbf{m}'_{x,J}\right\| \nonumber\\
    &= Op(\|\mathcal{P_{{S_J}}}m'_x\|h),
\end{align}
     where $a^{-1} = 1 + op\left((n\psi_x(h))^{-1/2}\right)$ by Ba\'illo and Gran\'e (2009 \cite{baillo2009})

Similarly, the bound of $T_1$ is:
\begin{align*}
    T_1 &= \left[A_{11} \quad \mathbf{A_{12}}\right] \dfrac{1}{n} \mathbf{C}_x^T \mathbf{\Delta} \begin{pmatrix} \mathcal{P_{{S_J}^\perp}} m'_x (X_1 -x) \\ \vdots \\ \mathcal{P_{{S_J}^\perp}} m'_x (X_n -x) \end{pmatrix} \nonumber\\
    &= \left[A_{11} \quad \mathbf{A}_{12}\right] \begin{bmatrix} \dfrac{1}{n}\sum_i \Delta_i \mathcal{P_{{S_J}^\perp}} m'_x (X_i -x) \\ \dfrac{1}{n} \sum_i \Delta_i \mathcal{P_{{S_J}^\perp}} m'_x (X_1 -x) \mathbf{c}_i \end{bmatrix}
\end{align*}
Based on Lemma \ref{suplemma:1}, $\dfrac{1}{n}\sum_i \Delta_i \mathcal{P_{{S_J}^\perp}} m'_x (X_i -x) = Op\left(\left\|\mathcal{P_{{S_J}^\perp}} m'_x\right\| h\right)$, and $\left\|\dfrac{1}{n} \sum_i \Delta_i \mathcal{P_{{S_J}^\perp}} m'_x (X_1 -x) \mathbf{c}_i\right\| = Op\left(\left\|\mathcal{P_{{S_J}^\perp}} m'_x\right\| h^2\right)$. Thus, 
\begin{align} \label{supeq:T1}
    T_1 &= a \cdot Op\left(\left\|\mathcal{P_{{S_J}^\perp}} m'_x\right\| h\right) + a^2\cdot \max_{1\le j \le J} \dfrac{1}{\tilde{\gamma}_j+\lambda_j}Op\left(\left\|\mathcal{P_{{S_J}^\perp}} m'_x\right\| h^3\right) \nonumber\\
    &= Op\left(\left\|\mathcal{P_{{S_J}^\perp}} m'_x\right\| h\right) + \kappa_J Op\left(\left\|\mathcal{P_{{S_J}^\perp}} m'_x\right\| h^3\right)
\end{align}
And $T_2$ is consequently 
\begin{align} \label{supeq:T2}
    T_2 &= \left[A_{11} \quad \mathbf{A}_{12}\right] \begin{bmatrix} \dfrac{1}{n} \sum_i \dfrac{1}{2}m''_{r_i}\left(\left(X_i - x\right)^2\right) \\ \dfrac{1}{n} \sum_i \dfrac{1}{2}m''_{r_i}\left(\left(X_i - x\right)^2\right) \Delta_i \mathbf{c}_i\end{bmatrix} \nonumber\\
    &= Op\left(h^2\right) + \kappa_J Op\left(h^4\right)
\end{align}

\subsection{FLLR-r Estimator Variance}
Variance of $\hat{m}\left(x\right)$ is
\begin{align} \label{supeq:var}
    \text{Var}\left(\hat{m}\left(x\right)\right) &= \dfrac{\sigma_e^2}{n} \mathbf{e}_1^T\left(\dfrac{1}{n} \mathbf{C}_x^T \mathbf{\Delta} \mathbf{C}_x+\mathbf{H}\right)^{-1}\dfrac{1}{n}\mathbf{C}_x^T\mathbf{\Delta}^2\mathbf{C}_x\left(\dfrac{1}{n} \mathbf{C}_x^T \mathbf{\Delta} \mathbf{C}_x+\mathbf{H}\right)^{-1}\mathbf{e}_1 \nonumber\\
    &= \dfrac{\sigma_e^2}{n}\left[A_{11} \quad \mathbf{A}_{12}\right]\begin{bmatrix} \dfrac{1}{n}\sum_i \Delta_i^2 & \dfrac{1}{n}\sum_i \Delta_i^2 \mathbf{c}_i^T\\ \dfrac{1}{n}\sum_i \Delta_i^2 \mathbf{c}_i & \dfrac{1}{n}\mathbf{C}_p^T\mathbf{\Delta}^2\mathbf{C}_p \end{bmatrix} \begin{bmatrix}A_{11} \\ \mathbf{A}_{21}\end{bmatrix} \nonumber\\
    &= \dfrac{\sigma_e^2}{n}\left[\left(\dfrac{1}{n}\sum_i \Delta_i^2\right)A_{11}^2 + 2A_{11}\left(\dfrac{1}{n}\sum_i \Delta_i^2 \mathbf{c}_i^T\right)\mathbf{A}_{21} + \dfrac{1}{n} \mathbf{A}_{12} \mathbf{C}_p^T\mathbf{\Delta}^2\mathbf{C}_p\mathbf{A}_{21}\right]
\end{align}

\begin{itemize}
    \item Bound of $A_{11}$ is:
    \begin{align*}
        A_{11} &= a + a^2 \dfrac{1}{n} \mathbf{1}^T\mathbf{\Delta}\mathbf{C}_p \mathbf{V}\mathbf{D}^{-1}\mathbf{V}^T\dfrac{1}{n}\mathbf{C}_p^T\mathbf{\Delta}\mathbf{1}\\
        &\le a + a^2 \left\|\dfrac{1}{n} \mathbf{1}^T\mathbf{\Delta}\mathbf{C}_p \mathbf{V}\right\|^2 \left\|\mathbf{D}^{-1}\right\| \\
        &= a + a^2 Op\left(h^2\right) \max_{1 \le j \le J}\dfrac{1}{\tilde{\gamma}_j +\lambda_j}
    \end{align*}
    \item
    Bound of the norm of $\mathbf{A}_{12}$ is:
    \begin{align*}
        \left\|\mathbf{A}_{12}\right\| &= a\left\|\dfrac{1}{n}\mathbf{1}^T\mathbf{\Delta}\mathbf{C}_p \mathbf{V}\mathbf{D}^{-1}\mathbf{V}^T\right\|\\
        &\le a\left\|\dfrac{1}{n}\mathbf{1}^T\mathbf{\Delta}\mathbf{C}_p \mathbf{V}\right\| \left\|\mathbf{D}^{-1}\right\|\\
        &=a\cdot \max_j\dfrac{1}{\tilde{\gamma}_j+\lambda_j} Op\left(h\right).
    \end{align*}
\end{itemize}

Then together with Lemma \ref{suplemma:1}, variance of $\hat{m}(x)$ can be bounded by
\begin{align}
    \text{Var}\left(\hat{m}\left(x\right)\right) &= \dfrac{\sigma_e^2}{n}\left[\left(\dfrac{1}{n}\sum_i \Delta_i^2\right)A_{11}^2 + 2A_{11}\left(\dfrac{1}{n}\sum_i \Delta_i^2 \mathbf{c}_i^T\right)\mathbf{A}_{21} +  \mathbf{A}_{12} \dfrac{1}{n}\mathbf{C}_p^T\mathbf{\Delta}^2\mathbf{C}_p\mathbf{A}_{21}\right] \nonumber\\
    &= \dfrac{\sigma_e^2}{n}\left\{Op\left(\psi_x^{-1}(h)\right) + \kappa_J Op\left(\dfrac{h^2}{\psi_x(h)}\right) + \kappa_J^2 Op\left(\dfrac{h^4}{\psi_x(h)}\right)\right\} \nonumber\\
    &= Op\left(\dfrac{1}{n\psi_x(h)}\right) +\kappa_J  Op\left(\dfrac{h^2}{n\psi_x(h)}\right) + \kappa_J^2 Op\left(\dfrac{h^4}{n\psi_x(h)}\right)
\end{align}

Finally, with $h^2\max_j\dfrac{1}{\tilde{\gamma}_j+\lambda_j} = Op(1)$ as stated in assumption, bias of $\hat{m}(x)$ is $Op(\|\mathcal{P_{{S_J}}}m'_x\|h) + Op\left(\left\|\mathcal{P_{{S_J}^\perp}} m'_x\right\| h\right) + Op\left(h^2\right)$, and its variance is $Op\left(\dfrac{1}{n\psi_x(h)}\right)$.

\section{Derivative Estimation}
By Riesz representation, $m'_x(u)= \langle \sum_{j=1}^{\infty} \langle m'_x, \phi_{j} \rangle \phi_j, u\rangle$ for $u \in \mathcal{L}^2(\mathcal{T})$. To examine behaviors of derivative estimation by FLLR and FLLR-r, we first calculate the true derivatives $m'_{X_i}$ for training cases $X_i$, where $\langle m'_{X_i}, \phi_{j} \rangle = m'_{X_i}\left(\phi_{j}\right)$. For regression described in Eq.(12) of the main draft, derivative at basis index $j^*$ is:

\begin{equation*}
    m'_{X_i}\left(\phi_{j^*}\right) = 
    \begin{cases}
    \left(1-a\right) + a \left[-2\sqrt{\theta_{j^*}}U_{ij^*}\exp \left(-\theta_{j^*}U_{ij^*}^2\right)\right], \text{\qquad when } j^* \le 20,\\
    \left(1-a\right), \text{\qquad when } 20 < j^* \le 30, \\
    0, \text{\qquad when } j^* > 30.
    \end{cases}
\end{equation*}
Then the true derivative operator on cut-off $J^*$ basis, with $J^*$ selected by cross validation, is represented as $m'_{X_i, J^*} = \sum_{j=1}^{J^*} m'_{X_i}\left(\phi_{j}\right) \phi_j$. FLLR and FLLR-r estimated derivatives take the form $\hat{m}'_{X_i, J^*} = \sum_{j=1}^{J^*} \hat{m}'_{X_i}\left(\hat{\phi}_{j}\right) \hat{\phi}_j$, where $\hat{\phi}_j$ is estimated basis function from local-linear pre-smoothed raw data curves, and $\hat{m}'_{X_i}\left(\hat{\phi}_{j}\right)$ is derived by each method following implementation discussed in Section 4.2.  

\begin{figure}[h!] 
\centering
      \includegraphics[scale=0.37]{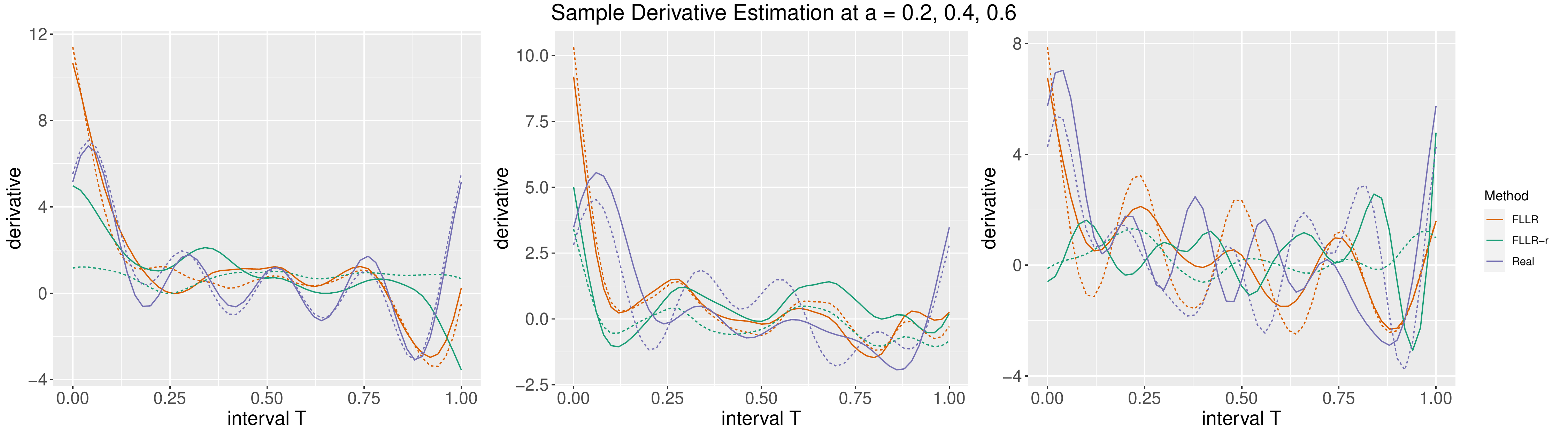}
  \caption{Estimated derivatives and their true counterparts at $a=0.2, 0.4, 0.6$. For each $a$, two samples are randomly selected, labeled by solid and dashed lines. FLLR and FLLR-r estimated derivatives and the true derivative at truncated basis $J^*$ are plotted for each sample, with types differentiated by color}
  \label{fig:est_deriv}
\end{figure}

Fig.\ref{fig:est_deriv} includes estimated derivatives and their true counterparts at $a=0.2, 0.4, 0.6$, to show the change in derivative estimation behaviors of FLLR and FLLR-r. The same generated sample data are used as in Section 4.3. For each $a$, two samples are randomly selected, labeled by solid and dashed lines. Then FLLR and FLLR-r estimated derivatives and the true derivative at truncated basis $J^*$ are plotted for each sample, with different derivative types differentiated by color. Considering the display quality, we use only two samples for each $a$ level, but the estimator trend is still apparent: when $a$ is small, FLLR and FLLR-r are able to approximate derivatives. However, with $a$ increasing, both estimators tend to underfit the derivatives as the regression operator becomes more nonlinear. Derivatives by FLLR-r in such situation are generally flatter than those by FLLR, due to the effect of ridge penalty. Table \ref{deriverr} records the averaged squared error from the two randomly selected samples by each derivative estimator, under different $a$ levels.

\begin{table}[h!]
\centering
\begin{tabular}{r|r|r}
  \hline
 & err(FLLR) & err(FLLR-r) \\ 
  \hline
$a = 0.2$ & 3.29 & 5.11 \\ 
\hline

$a = 0.4$ & 3.26 & 3.94 \\ 
   \hline
$a = 0.6$ & 3.58 & 6.01\\
\hline
\end{tabular}
\caption{Averaged Squared error of the estimated derivatives from the true derivative functions at different $a$ levels by FLLR and FLLR-r.}
\label{deriverr}
\end{table}

\bibliographystyle{apalike}
\bibliography{references} 